\documentclass[10pt, conference, letterpaper]{IEEEtran}

\usepackage[pass]{geometry}

\renewcommand{\baselinestretch}{0.97}

\usepackage{amsfonts}
\usepackage{amsthm}
\usepackage{amsmath}
\usepackage{amssymb}
\usepackage{array}
\usepackage{bm}         % for nice bold math symbols using \bm{...}
\usepackage{bbm}
\usepackage{cite}
\usepackage{dsfont}
\usepackage{epsfig}
\usepackage{epstopdf}
\usepackage{float}
\usepackage[T1]{fontenc}
\usepackage{graphicx}
\usepackage{indentfirst}
\usepackage{multicol}
\usepackage{blkarray}
\usepackage{makecell}

\usepackage{ifthen}
\newboolean{long} 

\usepackage[ruled,linesnumbered]{algorithm2e}
\SetAlFnt{\small}

\SetCommentSty{mycommfont}

\usepackage{algorithmicx}
\usepackage{algpseudocode}
\usepackage{comment}

\usepackage{array}
\newcolumntype{L}[1]{>{\raggedright\let\newline\\\arraybackslash\hspace{0pt}}m{#1}}
\newcolumntype{C}[1]{>{\centering\let\newline\\\arraybackslash\hspace{0pt}}m{#1}}
\newcolumntype{R}[1]{>{\raggedleft\let\newline\\\arraybackslash\hspace{0pt}}m{#1}}

\usepackage{subfigure}
\usepackage{url}
\usepackage{multirow}
\usepackage{color}
\usepackage{blkarray} % added for matrix
\usepackage{tikz} % added to color  matrix element
\usetikzlibrary{calc} % added to color  matrix element
\usepackage[export]{adjustbox} % top alignment of subfigures
\usepackage{setspace} % for line spacing

\usepackage{todonotes} % disable todo notes: [disable]


\newtheorem{thm}{Theorem}
\newtheorem{lem}{Lemma}
\newtheorem{mydef}{Definition}

\newtheorem{obs}{Observation}
\newcommand{\set}[1]{\mathcal{#1}} % sets
\newcommand{\eqdef}{\triangleq} % definition
\newcommand{\mat}[1]{\mathsf{#1}} % matrices (capital greek symbols work too)
\newcommand{\schd}[1]{\set{#1}} % a schedule notation
\newcommand{\vect}[1]{\bm{#1}} % vector
\newcommand{\card}[1]{\left|#1\right|} % cardinality of a set.
\newcommand{\trans}[1]{#1^{\textup{\textsf{\tiny T}}}} % transpose

\newfont{\mycrnotice}{ptmr8t at 7pt}
\newfont{\myconfname}{ptmri8t at 7pt}

\IEEEoverridecommandlockouts

\begin{document}
% \sloppy

\def\sharedaffiliation{%
\end{tabular}
\begin{tabular}{c}}

\renewcommand{\vec}[1]{\mathbf{#1}}

%\special{papersize=8.5in,11in}
%\setlength{\pdfpageheight}{\paperheight}
%\setlength{\pdfpagewidth}{\paperwidth}

%\title{Opportunistic cellular communications\\with clusters of dual-radio mobiles}
%\title{Opportunistic cluster scheduling in cellular networks with WLAN-assisted relay}
%\title{Boosting Throughput, Fairness and Energy Efficiency via Opportunistic Scheduling and User Cooperation in Cellular and Ad-hoc Networks}
%\permission{
%}

\title{Optimal Joint Routing and Scheduling in Millimeter-Wave Cellular Networks}

\author{
\IEEEauthorblockN{Dingwen Yuan\IEEEauthorrefmark{1},
    Hsuan-Yin Lin\IEEEauthorrefmark{2},
    J\"{o}rg Widmer\IEEEauthorrefmark{3} and
    Matthias Hollick\IEEEauthorrefmark{1} } 
\begin{tabular}{*{3}{>{\centering}c}}
\IEEEauthorrefmark{1}SEEMOO, Technische Universit\"{a}t Darmstadt & \IEEEauthorrefmark{2}Simula@UiB, Bergen Norway & \IEEEauthorrefmark{3}Institute IMDEA Networks \tabularnewline
\url{{firstname.lastname}@seemoo.tu-darmstadt.de} & \url{hsuan-yin.lin@ieee.org} & \url{joerg.widmer@imdea.org}
\end{tabular}
}

\maketitle

%\numberofauthors{3}
%% Three authors sharing the same affiliation.
%    \author{
%      \alignauthor Arash Asadi\\
%      \affaddr{Institute IMDEA Networks}\\
%      \affaddr{University Carlos III of Madrid}\\
%      \affaddr{Leganes (Madrid), Spain}\\
%      \affaddr{arash.asadi@imdea.org}
%%
%      \alignauthor Vincenzo Mancuso\\
%      \affaddr{Institute IMDEA Networks}\\
%      \affaddr{University Carlos III of Madrid}  \\
%     \affaddr{Leganes (Madrid), Spain}\\
%     \affaddr{vincenzo.mancuso@imdea.org}
%%
%      \alignauthor Rohit Gupta\\
%      \affaddr{Eurecom}\\
%      \affaddr{Sophia Antipolis, France}\\
%      \affaddr{Rohit.Gupta@eurecom.fr}
%%
%                }

%\author{
%\alignauthor
%Arash Asadi$^{\dagger\ddagger}$, Peter Jacko$^{*}$ , and Vincenzo Mancuso$^{\dagger\ddagger}$\\
%	%\titlenote{Dr.~Trovato insisted his name be first.}\\
%%       \affaddr{and, Spain}\\
%%       \affaddr{Avda. del Mar Mediterraneo, 22}\\
%%       \affaddr{Leganes (Madrid), Spain}\\
%      \affaddr{arash.asadi@imdea.org,\qquad p.jacko@lancaster.ac.uk,\qquad and vincenzo.mancuso@imdea.org}\\
%      \affaddr{IMDEA Networks Institute$^{\dagger}$,\qquad University Carlos III of Madrid$^{\ddagger}$,\qquad Lancaster University$^{*}$ }\\
%}
%

\setboolean{long}{true}  

\begin{abstract}
Millimeter-wave (mmWave) communication is a promising technology to cope with the expected exponential increase in data traffic in 5G networks. mmWave networks typically require a very dense deployment of mmWave base stations (mmBS).
To reduce cost and increase flexibility, wireless backhauling is needed to connect the mmBSs. The characteristics of mmWave communication, and specifically its high directionality, imply new requirements for efficient routing and scheduling paradigms. 
We propose an 
%on-average polynomial time 
efficient scheduling method, so-called {\em schedule-oriented optimization}, based on matching theory that optimizes QoS metrics jointly with routing. It is capable of solving any scheduling problem that can be formulated as a linear program whose variables are link times and QoS metrics.
As an example of the schedule-oriented optimization, we show the optimal solution of the maximum throughput fair scheduling (MTFS). Practically, the optimal scheduling can be obtained even for networks with over 200 mmBSs. To further increase the runtime performance, we propose an efficient edge-coloring based approximation algorithm with provable performance bound. It achieves over 80\% of the optimal max-min throughput and runs 5 to 100 times faster than the optimal algorithm in practice.
Finally, we extend the optimal and approximation algorithms for the cases of multi-RF-chain mmBSs and integrated backhaul and access networks.
\end{abstract}

\section{Introduction}
5G cellular systems are embracing millimeter wave (mmWave) communication in the $10$-$300$ GHz band where abundant bandwidth is available to achieve Gbps data rates. One of the main challenges for mmWave systems is the high propagation loss at these frequency bands. 
Although it can be partially compensated by directional antennas~\cite{Rappaport:2013jk, Rangan:2014ia}, the effective communication range of a mmWave base station (mmBS) remains around $100$ meters at best. Thus, base station deployment density in 5G will be significantly higher than in 4G~\cite{Singh:2015eh, ghosh2014millimeter}. This leads to high infrastructure cost for the operators. Besides the cost of site lease, backhaul link provisioning is in fact the main contributor to this cost because the mmWave access network may require multi-Gbps backhaul links to the core network. Currently, such a high data rate can only be accommodated by fiber-optic links which have high installation cost and are inflexible with respect to relocation.  

%{\bf Motivation}. 
Recent studies show that mmWave self-backhauling is a  cost-effective alternative to the wired backhauling. This approach is particularly interesting in a heterogeneous network setting where the existing cellular base stations (eNBs) act as a gateway for the mmBSs. Fig.~\ref{fig:net-model} illustrates such a setup in which the eNB can reach mmBSs directly or via other mmBSs.
%We also can observe that directionality of mmWave communication reduces the intercell interference considerably. 
Moreover, directionality of mmWave communication reduces or removes the wireless backhaul interference and allows simultaneous transmissions of multiple links over the same channel as long as their beams do not overlap. However, the number of simultaneous links a base station can have is limited by the number of its RF chains. 
%Due to energy consumption, complexity and manufacturing costs, the number of RF chain are limited and so is the number of links. 

\begin{figure}[t!]
\centering
		\includegraphics[width=0.5\textwidth]{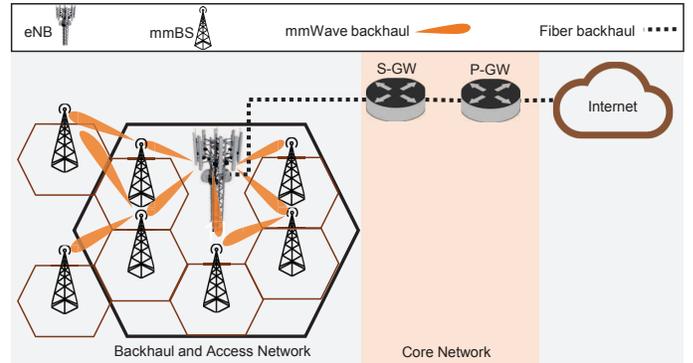} %[width=70mm,angle=0] [scale=0.80, angle=0]
		\caption{mmWave self-backhauling setup.}
		\label{fig:net-model}
\end{figure}

%{\bf Challenge.} 
To date, much of the research on mmWave communication has been dedicated to issues faced by the mobile users (UEs) in the access networks. 
How to maximize performance such as throughput and energy efficiency in mmWave backhaul and access networks has received less attention. Here, two important issues need to be addressed: (i) routes to be taken, (ii) the scheduling of the transmission overs the links.

 %   are scheduled and Given the mmWave self-backhauling resembles a multihop network with negligible interference between nodes and capablility 

%{\bf Solution.} 
%Due to the directional communcation, the backhaul links have negligible mutual interference. 
A naive scheduling which lets the eNB serve all the mmBSs in a round robin fashion is neither practical nor efficient.
%, because (i) some mmBSs may not be directly connected with the eNB, and hence will starve, and (ii) 
If mmBSs' links to the eNB are weak compared to their links to other nearby mmBSs (which in turn have high-capacity links to the eNB), a schedule allowing multi-hop routing is much more favorable since it alleviates the bottleneck at the eNB. At the same time, the limited interference at mmWaves makes it efficient to maximize spatial reuse and operate as many links simultaneously as possible. 
The goal of the paper is to design a scheduler that exploits these characteristics to optimize mmWave backhaul efficiency.

The paper is organized as follows. We discuss related work in Sec.~\ref{s:related}. Sec.~\ref{s:sysmodel} provides the system model. The relation between a schedule and matchings is studied in Sec.~\ref{s:preliminary}. In Sec.~\ref{s:sched-fair}, we present our optimal  schedule-oriented optimization method, through an example---maximum throughput fair scheduling. We then propose a fast edge-coloring based approximation algorithm in Sec.~\ref{s:approx-algo}. In Sec.~\ref{s:extend}, we extend our algorithms to more general scenarios. Sec.~\ref{s:eval} shows the numerical evaluation and Sec.~\ref{s:conclusions} concludes the paper.

\section{Related Work}
\label{s:related}
Few works on mmWave backhaul and access network scheduling exist~\cite{Niu15, Zhu16, Feng16, Li17}. These works share the assumptions that (i) the traffic demand is measured in discrete units of slots or packets, and (ii) a flow has to be scheduled sequentially, i.e., a hop closer to the source should be scheduled earlier than a hop farther away. The resulting optimization problems are all formulated as mixed integer programming (MIP) problems. As MIPs are in general NP-complete, optimal solutions can only be computed for small networks with a few nodes. For practical use, these works all rely on heuristics, which are based on the ideas such as greedy edge coloring~\cite{Niu15, Feng16} or finding the maximum independent set in a graph~\cite{Zhu16, Li17}. Furthermore, \cite{Niu15, Zhu16, Li17} assume that routing is pre-determined, which does not fully exploit the freedom given by a reconfigurable mmWave backhaul, and may limit performance.

In contrast, our work relaxes the constraint of sequential flow scheduling (i.e., if needed, packets are queued for a short time) which does not harm the long-term throughput, and allows the slots in a schedule to be of any length. Based on these assumptions, we propose a polynomial time optimal scheduling method which is shown by simulation to be practical for mmWave cellular networks. Moreover, the scheduling takes QoS optimization goals or QoS requirements as input and finds an optimal routing automatically. The first attempt to solve the problem of joint routing and scheduling in a network with Edmonds' matching formulation goes back to~\cite{Hajek88}. Hajek et al.'s polynomial time scheduling algorithm is different from ours in that it minimizes the schedule length. Furthermore, we use a one-step {\em schedule-oriented} approach while they first compute the optimal link time and then compute the minimum length schedule given the link time.

Following~\cite{Hajek88}, recent research on scheduling focuses on optimization of delay~\cite{Huang13} and queue length~\cite{Angelakis14}, as well as investigating more realistic interference models~\cite{HariharanS12}. 
Another interesting line of research is the on-line node-based scheduling algorithms which achieve good  performance bound in throughput and evacuation time~\cite{Ji16}.

\section{System model}
\label{s:sysmodel}
The system model considers a backhaul network which has a single eNB, equipped with multiple mmWave RF chains, and multiple mmBSs, each equipped with a single RF chain. 
Later in Sec.~\ref{s:extend}, we will show that our optimization method applies equally to (i) the case where each node has multiple RF chains, and (ii) a network model that includes UEs.

We consider an eNB macro cell together with a number of mmWave base stations in a heterogeneous TDMA cellular network. The eNB acts as the backhaul gateway for mmBSs. In addition to an LTE radio interface, the eNB is equipped with $ R $ mmWave RF chains. There are $W$ single-RF-chain mmBSs in the macro cell. We assume analog or hybrid beamforming with $R$ RF chains which allows up to $R$ simultaneous links at the eNB.
%\footnote{We assume analog beamforming for the simplicity of presentation because the number of maximum simultaneous links equals the RF chain number.} 
We use directed graphs to model the links between the different nodes in the network. Fig.~\ref{fig:sys_model} illustrates a toy example of a backhaul network. The following analysis focuses on downlink communication. The same analysis can be applied to the uplink scenario. 
%\blue{In this paper, we refer to a set with and a vector with  and scalar values with .... . }

Let $G = (\set{V} ,\set{E})$ be a directed graph with the vertex (node) set $\set{V}$ and edge (link) set $\set{E}$. Each edge represents a potential link between two vertices. The capacity of each edge $e$ is denoted $c_e$. 
%$ c _ e = b \log_2( 1 + \gamma  ) $, where $ b $ denotes the bandwidth and $\gamma $ is the signal-to-noise ratio (SNR). 
The received power is given by $p _ { \text{rx} }  = p _ { \text{tx} } + g _ { \text{x} } - PL$, where $ p_ { \text{tx} } $ is the transmission power, $ g_{ \text{x} } $ is the directivity gain, and $ PL $ is path loss between the transmitter and the receiver. % which is given by \eqref{eq:pathloss}.  
% \begin{align}
$PL ( d )  = \alpha + 10 \beta \log _ { 10 } { d } + \xi$,
%\label{eq:pathloss}
% \end{align}
where  $ \alpha $, $\beta$ are constants that depend on the frequency and line-of-sight conditions. $d$ is the distance between the transmitter and receiver. $\xi$ represents the shadowing effect and is a normal distributed random variable with zero mean and $\sigma$ standard deviation. 

\begin{figure}[!t]
\centering
	\includegraphics[width=5cm]{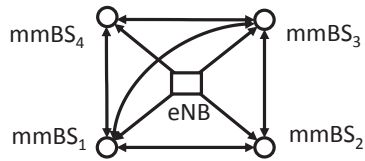}
\caption{A backhaul network example. The edges shown are the potential links for downlink schedule. A bidirectional edge represents two opposite links.}
\label{fig:sys_model}	
\end{figure}

We can observe in Fig.~\ref{fig:sys_model} that there are many ways to schedule downlink communication among the eNB and mmBSs. 
%In this paper, we leverage graph theory to model  the above mmWave backhaul system. We optimize this model for given QoS constraints via linear optimization. 
Our goal is to obtain the optimal {\em unit length schedule} with respect to a QoS metric, while satisfying given QoS requirements and the  constraints on simultaneous transmissions. In practice, the unit time is the duration of the radio frame. We observe that each feasible schedule $ \set{S} $ can always be divided into  $ N \geq 1$ slots numbered as 1 to $N$. We define $ t_i $ as the length of the $i$-th slot. It is required that $ \sum _ { i = 1 } ^ N t_i = 1 $ and $t_i > 0$.  Moreover, in the $i$-th slot, a set of links $\set{E}_i  \subseteq \set{E}$ (can be empty) are active for the whole slot.

\section{Preliminary: schedule polyhedron}
\label{s:preliminary}
This section first shows the relation between a feasible schedule and matchings in a graph. Based on the relation, we mathematically formulate the set of all feasible schedules as the {\em schedule polyhedron} that is described by linear constraints.

Suppose that a set of links $\set{E}_i$ are scheduled in the $i$-th slot, and $e_1, e_2 \in \set{E}_i$ are two different links. Then $e_1$ and $e_2$ can not share a common mmBS node since a mmBS has one RF chain and is therefore half-duplex. On the other hand, $e_1$ and $e_2$ may share the eNB node given that $R > 1$. However, the number of links in $\set{E}_i$ that are incident to the eNB cannot be more than $R$. We enforce this constraint through the eNB expansion. 

\subsection{eNB expansion}\label{ss:e-v}
In graph $G$, we replace the eNB with $R$ {\em expanded eNBs}: $\text{eNB}_1,...,\text{eNB}_R$.  If the eNB is connected to a set of mmBSs, then each $\text{eNB}_i$ is connected to the same set of mmBSs with the same respective link capacities as the eNB. The resulting graph is equivalent to the original graph with respect to scheduling. Yet each expanded eNB has one RF chain. In the rest of the paper, the graph $G = (\set{V}, \set{E})$ is assumed to be expanded if not explicitly stated otherwise. An example of the process described above is shown in Fig.~\ref{fig:veNB}. Let $\set{R}$ and $\set{W}$ denote the set of expanded eNBs and mmBSs, respectively. 
%We assume without loss of generality that $R \le L$, where $L$ is the number of mmBSs that are connected to the eNB, since obviously $L$ RF chains at eNB are enough for $L$ half-duplex mmBSs.

As a result, we can ensure that the active link set of a slot corresponds to a {\em matching} in the expanded graph. 
A matching in a graph is defined as a set of edges in the graph that share no common vertices.  

\begin{figure}[!t]
\centering
	\includegraphics[width=0.5\textwidth]{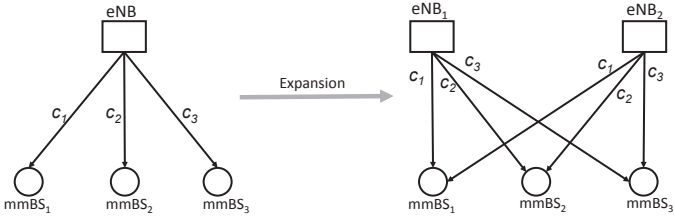}
\caption{eNB expansion with two RF chains.}
\label{fig:veNB}	
\end{figure}

\subsection{The schedule polyhedron}
We define the {\em link time} $t_e \in [0, 1]$ as the total active time of a link $e$ in a schedule. Correspondingly, $\vect{t}$ is the {\em link time vector}, each element of which is a link time $t_e, \forall e \in \set{E}$. 
A link time vector $\vect{t}$ is {\em feasible} if $\vect{t}$ can be scheduled in unit time.
We first define the {\em schedule polyhedron} $P$ and then prove that each point in $P$ is one-to-one mapped to each feasible link time vector.

\begin{mydef}[Schedule Polyhedron] Given a graph $G = (\set{V}, \set{E})$, the schedule polyhedron of $G$ is defined as the set of link time vectors $\vect{t}$ that satisfy the  following linear constraints.
\begin{IEEEeqnarray}{rCll}
    \IEEEyesnumber\label{eq:link-active-time}\IEEEyessubnumber*
    \sum_{e\in\delta(v)}t_e & \le &1 &\qquad \forall\, v \in \set{V},
    \label{eq:node-matching}\\
    \sum_{e\in \set{E}(\set{O})}t_e& \le &
    \left\lfloor\frac{|\set{O}|}{2}\right\rfloor &\qquad
    \forall\,\textnormal{ odd set }\set{O} \subseteq
    \set{V},\label{eq:oddset} \\
    t_e& \ge &0  &\qquad \forall\,e \in \set{E},
  \end{IEEEeqnarray}
where $\delta(v)$ is the set of links incident to node $v$. An {\em odd set} $\set{O}$ has odd number of nodes. $\set{E}(\set{O})$ is the set of links whose endpoints are both contained in $\set{O}$. 
\end{mydef}

The following lemma summarizes the relation between the schedule polyhedron and the unit length schedules.

\begin{lem} (1) Each point in the schedule polyhedron $P$ is a feasible link time vector $\vect{t}$, and (2) each feasible link time vector $\vect{t}$ is a point in $P$.
\label{lm:schedule-matching}
\end{lem}
\begin{proof}
The proof uses the Edmonds' matching polyhedron theorem~\cite{Edmonds65b}. A feasible schedule $\schd{S}$ consists of $N \ge 1$ slots. Each slot contains a set of links from $G$ that is a matching and therefore corresponds to a vertex of the matching polyhedron $Q$. Since $Q$ has the same formulation as $P$, except the variables are binary, $P$ and $Q$ has the same set of vertices. Since $\schd{S}$ has a length of 1, the link time vector $\vect{t}$ of $\schd{S}$ is a convex combination of the vertices of $P$. So it is a point in $P$. On the other hand, a point in $P$ can be written as a convex combination of all vertices of $P$ where each vertex corresponds to a slot. Therefore, each point in $P$ corresponds to a feasible unit time schedule. See the details in 
\ifthenelse{\boolean{long}}{Appendix~\ref{sec:proof_lemma1}}{the extended version~\cite{extended}}.  
\end{proof}

\begin{comment}
\subsection{Special case of no links between mmBSs} 
If the system model allows no links between mmBSs, then $G$ is a bipartite graph (the set of expanded eNBs $\set{R}$ and the set of mmBSs $\set{W}$ are the two parts). For bipartite graphs, the odd set constraints \eqref{eq:oddset} are redundant for the schedule polyhedron since they are redundant for the matching polyhedron~\cite{Birkhoff46}. 
By removing the constraints \eqref{eq:oddset}, whose number is exponential to $|\set{V}|$, from \eqref{eq:link-active-time}, we can improve the runtime efficiency of the optimization algorithms.  
After the removal, Lemma~\ref{lm:schedule-matching} is still valid.

\end{comment}

%\input{sched-cr}
\section{Maximum Throughput Fair Scheduling} \label{s:sched-fair}
Having established the relation between a schedule and matchings, we now investigate the problem of {\em maximum throughput fair scheduling} (MTFS) for backhaul networks. 
The goal of the problem is to maximize the downlink {\em network throughput} under the condition that the {\em max-min fairness}~\cite{Tang06, Tassiulas02} in throughput is achieved at the mmBSs.
The MTFS problem serves as one example of our method for scheduling optimization in mmWave backhaul networks.

\begin{mydef}[Maximum Throughput Fair Schedule] Given a backhaul network $G$ and a unit time schedule $\schd{S}$, let the throughput vector of $\schd{S}$ be $\vect{h}^{\schd{S}} = [h^{\schd{S}}_v | v \in \set{W}]$, where $h^{\schd{S}}_v$ denotes the downlink throughput of an mmBS node $v$.%\footnote{since an mmBS node is not a source node, $h^{\schd{S}}_v \ge 0$.} 

\noindent (i) A feasible unit time schedule $\schd{S}_f$ is said to satisfy the
    max-min fairness criteria if
    $\min_{v\in\set{W}} h_v^{\schd{S}_f}\geq
    \min_{v\in\set{W}}h_v^\schd{S}$ for any feasible unit time schedule
    $\schd{S}$. Such $\min_{v\in\set{W}} h_v^{\schd{S}_f}$ is called the max-min throughput.
  
\noindent (ii) A feasible unit time schedule $\schd{S}^*$ is a solution of the MTFS
    problem if $\schd{S}^*$ has achieved the maximum network throughput
    $\sum_{v\in\set{W}}h^{\schd{S}_f}_v$ among all possible feasible
    unit time schedule $\schd{S}_f$ satisfying the max-min fairness criteria in (i).
\end{mydef}

In the following, we present our general optimization method---{\em schedule oriented optimization}.

\subsection{Schedule oriented optimization}
The schedule oriented optimization solves a linear optimization problem, the solution to which is directly the optimal schedule. For the mathematical formulation of the optimization problem, we construct the {\em node-matching matrix}.

\begin{mydef}[Node-Matching Matrix] Given a directed graph
  $G=(\set{V}, \set{E})$. Suppose the number of all possible matchings
  of $G$ is $K$. Then the node-matching matrix $\mat{A}=[a_{i,j}]$ is
  a $|\set{V}| \times K$ matrix, whose elements is defined as
  follows: $a_{i,j} = c$, if there is a link with capacity $c$ entering
  node $i$ in the $j$-th matching; $a_{i,j} = -c$, if there is a link
  with capacity $c$ leaving node $i$ in the $j$-th
  matching. Otherwise, $a_{i,j}=0$.
  \label{def:node-matching-matrix}
\end{mydef}
\begin{figure}[!t]
\begin{minipage}{0.15\textwidth}
\begin{figure}[H]
	\centering
	\includegraphics[width=3.8cm]{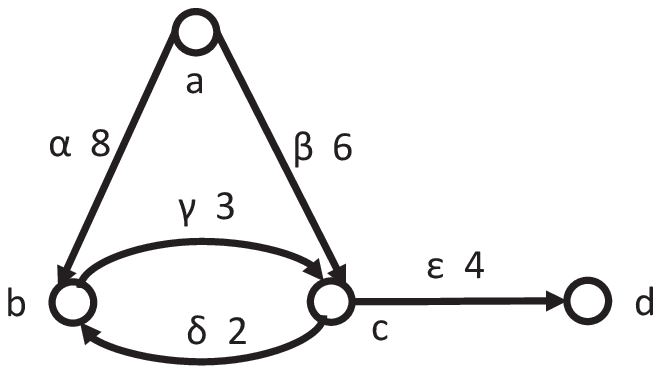}	
\end{figure}
\end{minipage} \quad \quad
\begin{minipage}{0.2\textwidth} \footnotesize
$
\begin{blockarray}{ccccccc}
   & \alpha & \theta & \gamma & \delta & \epsilon & \alpha,\epsilon \\
  \begin{block}{c[cccccc]}
    a  & -8 & -6 & 0  & 0  & 0  & -8\\
    b  & 8  & 0  & -3 & 2  & 0  & 8\\
    c  & 0  & 6  & 3  & -2 & -4 & -4\\
    d  & 0  & 0  & 0  & 0  & 4  & 4\\
  \end{block}
\end{blockarray}$
\end{minipage}
\caption{Node-matching matrix, $a$ to $d$ are nodes, $\alpha$ to $\epsilon$ are edges. The numbers are capacities.}
\label{fig:node-matching}
\end{figure}

As we will see, the node-matching matrix helps in formulating the throughput constraints at individual nodes.
Fig.~\ref{fig:node-matching} gives an example of node-matching matrix for a graph. 

Let $\mat{A}$ be the node-matching matrix of the backhaul network $G$, we define $\mat{A}^\set{W}$ as the submatrix of $\mat{A}$, which consists only of the rows of $\mat{A}$ related to the nodes in $\set{W}$ (mmBSs). 
As we pointed out in Sec.~\ref{s:preliminary}, the link set scheduled in each slot of a schedule must be a matching in $G$. We define $\vect{t}^\schd{S}$ as a $K\times 1$ {\em slot length vector}, each element of which is the length of a potential slot. Let the minimum throughput among all mmBSs be $\theta$. Then we can solve the MTFS problem in two steps: (i)
maximizing $\theta$ (such $\theta$ is the
max-min throughput) and (ii)
computing the optimal schedule $\schd{S}^*$ that offers the highest network throughput, given the max-min
throughput $\theta$.

\textbf{Linear programs for the MTFS problem.} 
The linear program to
maximize $\theta$ for step (i) of the MTFS problem can be formulated as
follows,
\begin{IEEEeqnarray}{lrCl}
  \IEEEyesnumber\label{eq:mtf-theta}\IEEEyessubnumber*
  \textnormal{maximize} & \theta & &
  \\
  \textnormal{subject to}\quad& \mat{A}^{\set{W}}\vect{t}^\schd{S}& \geq &\theta\vect{1}\label{eq:mtf-theta1}
  \\
  & \trans{\vect{1}}\vect{t}^\schd{S}& = &1\label{eq:mtf-theta2}
  \textnormal{ and }
   \vect{t}^\schd{S} \geq \vect{0},
\end{IEEEeqnarray}
where $\vect{1}$ and $\vect{0}$ represent the all-one and all-zero
column vector, respectively. The superscript ``$\trans{}$" denotes the
vector transposition. \eqref{eq:mtf-theta1} is the constraint that the
throughput at each mmBS should be at least
$\theta$. \eqref{eq:mtf-theta2} is the constraint that the schedule
length should be unit time. The feasibility of the schedule is
implicitly guaranteed by the formulation in terms of all possible
matchings.
 
After we have computed $\theta$ from \eqref{eq:mtf-theta}, we can formulate the linear program that maximizes the {\em network throughput}, i.e., the total throughput of all mmBSs under the condition that each mmBS has throughput at least $\theta$. 
\begin{IEEEeqnarray}{llCl}
  \IEEEyesnumber\label{eq:mtf}\IEEEyessubnumber*
  \textnormal{maximize} &\trans{\vect{c}}\vect{t}^\schd{S} & &
  \\
  \textnormal{subject to}\quad&\textnormal{\eqref{eq:mtf-theta1}, and \eqref{eq:mtf-theta2}.}
\end{IEEEeqnarray}
Here, $\vect{c}$ is the {\em capacity vector} whose element $c_j$ is
the cumulative capacity of all eNB-to-mmBS links in the $j$-th
matching $M_j$, i.e.,
$c_j = \sum_{(v,v')\in M_j,v\in\set{R}} c_{(v,v')}$, where
$(v,v')$ denotes the link from node $v$ to node $v'$. 
Note that $\theta$ is a variable in
\eqref{eq:mtf-theta}, but it is a constant in \eqref{eq:mtf}.

The apparent difficulty in solving \eqref{eq:mtf-theta} and
\eqref{eq:mtf} is the huge number of elements in $\vect{t}^\schd{S}$
(same as the number of matchings in $G$, which is exponential to the
number of vertices in $G$). Yet, we will show that we can still solve it in
polynomial time.

\begin{thm}
\label{thm:MTFS-polynomial-time}
	The MTFS problem can be solved in polynomial time with the ellipsoid algorithm~\cite{Khachiyan80}.
\end{thm}
\begin{proof}
  This proof uses a similar technique to the proof of Theorem 2 
  in~\cite{Nemhauser91}, which states that the fractional edge
  coloring can be solved in polynomial time by the ellipsoid
  algorithm. See the details in 
\ifthenelse{\boolean{long}}{Appendix~\ref{sec:proof_MTFS-polynomial-time}}{the extended version~\cite{extended}}.  
  % ~\ref{sec:proof_MTFS-polynomial-time}.
\end{proof}

Although polynomial, in practice the ellipsoid algorithm almost always takes longer than the simplex algorithm. In the following, we propose algorithms based on the revised simplex algorithm~\cite{Dantzig55} which does not require the generation of all columns of $\mat{A}^{\set{W}}$. Conceptually, the algorithms first create a feasible schedule. Then in each iteration, to improve the optimization objective, we replace one slot in the schedule by another matching (a set of simultaneous links) while keeping the schedule feasible, until  the optimum is reached.
The maximum weighted matching algorithm~\cite{Edmonds65b} is used to choose the matching (column) to enter the basis (schedule). 

\subsection{Solving the MTFS problem}

To optimize $\theta$, we need an initial basic feasible solution to \eqref{eq:mtf-theta}. Suppose that the backhaul network $G$ is connected, otherwise there are mmBSs unreachable from the eNB. We perform a {\em breath-first-search} (BFS) starting from an arbitrary expanded eNB, say $\text{eNB}_1$. The result is a tree $T$ that spans $\text{eNB}_1$ and all mmBSs. $T$ has exactly $W$ edges. The initial schedule $\schd{S}_0$ is constructed as follows: $\schd{S}_0$ has $W$ slots, each of which contains one link in $T$. Moreover, it is required that the throughputs of all mmBSs  are the same and the schedule takes exactly unit time. It is obvious that the initial solution is unique. We convert the linear
program \eqref{eq:mtf-theta} to the standard form \eqref{eq:mtf-theta_std} by introducing $W$
\emph{surplus variables} $s_i$ as follows.
\begin{IEEEeqnarray}{lrCl}
  \IEEEyesnumber\label{eq:mtf-theta_std}\IEEEyessubnumber*
  \textnormal{minimize} & &\trans{\vect{f}}\vect{x}&
  \\
  \text{subject to}\quad& \mat{U}\vect{x}& = &\vect{g}
  \textnormal{ and }
   \vect{x} \ge \vect{0},
\end{IEEEeqnarray}
where
$
  \mat{U} \eqdef [\mat{U}^1|\mat{U}^2|\mat{U}^3]\eqdef
  \left[\begin{array}{c|c|c}
          \mat{A}^{\set{W}} & -\vect{1} &  -\mat{I}
          \\
          \trans{\vect{1}}  & 0 & \trans{\vect{0}}
        \end{array}\right],
$
$\trans{\vect{f}}=\bigl[\trans{\vect{0}} \,|\, {-1} \,|\, \trans{\vect{0}}\bigr]$,
$\trans{\vect{x}}\eqdef \bigl[\trans{(\vect{t}^\schd{S})}\,|\, \theta
\,|\,\trans{\vect{s}}\bigr]$, and
$\trans{\vect{g}}\eqdef\bigl[\trans{\vect{0}}\,|\,1\bigr]$. 
Alg.~\ref{alg:mtf-theta} shows the computation of the max-min throughput $\theta$.

\begin{algorithm}[!t]
\SetAlgoLined
Set the basis $\mat{B} = \mat{B}_0$ corresponding to the initial schedule $\schd{S}_0$\;
\While{True} {
	%Compute the feasible solution $\vect{x}_{\mat{B}} = \mat{B}^{-1} \vect{g}$\;
	Compute the dual variable $\trans{\vect{p}} = \trans{\vect{f}_{\mat{B}}} \mat{B}^{-1} $\;
	Set weight $w_{(v_i, v_j)}$ to each link $(v_i, v_j)$ of $G$ as follows.
\begin{equation*}
  w_{(v_i,v_j)}=
  \begin{cases}
    c_{(v_i,v_j)} (p_j - p_i) & \text{if }v_i, v_j \in \set{W}
    \\
    c_{(v_i,v_j)} p_j         & \text{otherwise }v_i\in\set{R},v_j\in\set{W}
  \end{cases} 
\end{equation*}
Do max weighted matching on $G$. Let the optimal matching be $M$, compute $\eta_1 = -\sum_{e \in M} w_e - p_{W+1}$\label{alg:mtf-a1}\;
Compute $\eta_2 = -1 + \sum_{k = 1}^{W}p_k$\label{alg:mtf-a2}\;
Compute $\eta_3 = \min_{1 \le k \le W} {p_k}$\label{alg:mtf-a3}\;
Compute $\eta = \min(\eta_1, \eta_2, \eta_3)$ and let the corresponding column be $\vect{u}_{\eta} \in \mat{U}$\;
\eIf {$\eta \ge 0$} {%$\mat{B}_{\theta} = \mat{B}$\;
\Return the optimal $\theta$ and $\mat{B}_{\theta} = \mat{B}$\;}
{Update $\mat{B}$ by replacing a column of $\mat{B}$ with $\vect{u}_{\eta}$ according to the simplex algorithm\;}
%Compute $\vect{u} = \mat{B}^{-1} \vect{u}_{\eta}$. $\ell = \argmin_{\{k|u_k > 0\}} \frac{x_{\mat{B}(k)}} {u_k}$. Replace the $\ell$-th column of $\mat{B}$, $\vect{b}_{\ell}$ with $\vect{u}_{\eta}$\;}
}
\caption{Compute the max-min throughput $\theta$}
\label{alg:mtf-theta}
\end{algorithm}

The basis $\mat{B}$ is a square matrix that consists of $W+1$ columns from $\mat{U}$. $\vect{f}_{\mat{B}}$ is the elements of $\vect{f}$ corresponding to the basis $\mat{B}$.
The lines \ref{alg:mtf-a1}, \ref{alg:mtf-a2} and \ref{alg:mtf-a3} compute the minimum reduced cost of a column in the matrices $\mat{U}^1, \mat{U}^2$ and $\mat{U}^3$ respectively. 
 To decrease $-\theta$, we need to find a column of $\mat{U}$, $\vect{u}_{k}$ that has negative reduced cost $f_k - \trans{\vect{p}} \vect{u}_k < 0$ to enter the basis. In the algorithm, we find the column $\vect{u}_{\eta}$ in $\mat{U}$ that produces the minimum reduced cost $\eta$. If $\eta \ge 0$, then no columns can be used to decrease $-\theta$, thus we have reached the optimum. 

Let the final basis in computing $\theta$ be $\mat{B}_{\theta}$. To directly use $\mat{B}_{\theta}$ as the initial basis to the solution of step (ii) of the MTFS problem, we add an artificial scalar variable $y \ge 0$ to \eqref{eq:mtf} and replace the constraint $\mat{A}^{\set{W}}  \vect{t}^{\schd{S}} \ge \theta \vect{1}$ with $\mat{A}^{\set{W}}  \vect{t}^{\schd{S}} - \vect{1} y \ge \theta \vect{1}$. 
Since $\theta$ is the max-min throughput, the feasible $y$ must be 0. Hence, the optimal solution (maximum network throughput) to \eqref{eq:mtf} is unaffected.
Again, we convert \eqref{eq:mtf} into the standard form of \eqref{eq:mtf-theta_std}, which is solvable with the revised simplex algorithm. 

In the standard form, $\mat{U}$ remains unchanged,
we redefine
$\trans{\vect{f}} \eqdef \bigl[-\trans{\vect{c}} \,|\, 0 \,|\,
\trans{\vect{0}}\bigr]$, 
$\trans{\vect{x}}\eqdef \bigl[\trans{(\vect{t}^\schd{S})}\,|\, y
\,|\,\trans{\vect{s}}\bigr]$, and
$\trans{\vect{g}}\eqdef\bigl[\theta\trans{\vect{1}}\,|\,1\bigr]$.
The optimization algorithm is similar to
Alg.~\ref{alg:mtf-theta} and is outlined in Alg.~\ref{alg:mtf}. 
%It keeps updating the matchings (slots) and the associated basis $\mat{B}$ of the unit time schedule until the optimum is reached. 
Since the basis $\mat{B}$ is a square matrix of $W+1$
dimension, it follows that the optimal schedule $\schd{S}^*$ contains
no more than $W+1$ slots. Additionally, since the links on a 
flow from the eNB to a destination mmBS may not be scheduled in sequential
order, some transmission opportunities of the flow in the first few frames may be wasted. Therefore, maximum throughput is achieved in the long-term.

\begin{algorithm}[!t]
\SetAlgoLined
Set the basis $\mat{B} = \mat{B}_{\theta}$\;
\While{True} {
	Compute the dual variable $\trans{\vect{p}} = \trans{\vect{f}_{\mat{B}}} \mat{B}^{-1}$\;
	Set weight $w_{(v_i, v_j)}$ to each link $(v_i, v_j)$ of $G$ as follows. 
\begin{equation*}
  w_{(v_i, v_j)} = 
  \begin{cases}
    c_{(v_i,v_j)} (p_j - p_i) & \text{if }v_i,v_j\in\set{W} \\
    c_{(v_i,v_j)} (p_j + 1)   & \text{otherwise }v_i\in\set{R},v_j\in\set{W}
  \end{cases}
\end{equation*}
 Do max weighted matching on $G$. Let the optimal matching be $M$, compute $\eta_1 = -\sum_{e \in M} w_e - p_{W+1}$\;
	$\eta_2 =  \sum_{k = 1}^{W}p_k$\;
	$\eta_3 = \min_{1 \le k \le \set{W}} {p_k}$\;
	Compute $\eta = \min(\eta_1, \eta_2, \eta_3)$ and let the corresponding column be $\vect{u}_{\eta} \in \mat{U}$\;
	\eIf {$\eta \ge 0$} 
	{\Return the optimal schedule $\schd{S}^*$ corresponding to $\mat{B}$\;}
	{Update $\mat{B}$ by replacing a column of $\mat{B}$ with $\vect{u}_{\eta}$\;}

}
\caption{Solving the MTFS problem}
\label{alg:mtf}
\end{algorithm}

\subsection{Generalization}
The scheduled-oriented optimization method illustrated by the optimal MTFS algorithm is quite general. 
It can solve any scheduling problem that can be formulated as a linear program whose variables are link times and QoS metrics.
For example, it can optimize for the constraint that each mmBS has a minimum throughput requirement.
%, as the constraint can be easily expressed with the help of the node-matching matrix $\mat{A}^{\set{W}}$. 
Another example is that the proposed method can optimize the energy consumption as it can be translated into the minimization of total transmission time in a schedule. We do not further elaborate on them due to the space limitation. Moreover, in Sec.~\ref{s:extend}, we extend the optimization method to backhaul and access networks, as well as to multi-RF chains at each node.

\section{Edge-coloring based approximation algorithm} \label{s:approx-algo}
In Sec.~\ref{s:sched-fair}, we proposed an optimal joint routing and scheduling algorithm for mmWave backhaul networks. Although it is optimal, it may have a high runtime (c.f. the evaluation in Sec.~\ref{s:eval}) when the number of mmBS nodes is large. Hence, we propose a run-time efficient {\em edge-coloring (EC) based approximation algorithm} that has a provable performance bound.
The EC algorithm follows a two-step approach of (i) computing the link time and (ii) scheduling within unit time. 

\subsection{Step (i): computing link time} \label{ss:ec-step1}
To precisely compute the link time, we need to include all the constraints of the schedule polyhedron \eqref{eq:link-active-time}.
Since the number of odd set constraints \eqref{eq:oddset} is huge, which leads to a high runtime for the optimization, instead we use a small set of constraints that is a necessary but not sufficient condition for a feasible unit time schedule. 
The selection of the new set of constraints is based on the following observation.

Let $G = (\set{V}, \set{E})$ be the backhaul network before the eNB expansion (Sec.~\ref{ss:e-v}) and $G^{\set{W}}$ be the subgraph of $G$, which contains only the mmBSs $\set{W}$ and the links among them. We define $\nu = \left \lfloor \frac{W}{2} \right \rfloor$ as an upper bound of the maximum number of mmBS-to-mmBS links that can be active simultaneously\footnote{The active links form a matching in $G^{\set{W}}$. According to the definition of matching, the number of active links is upper-bounded by  $\bigl\lfloor\frac{W}{2}\bigr\rfloor$.}.
%The precise number of $\nu$ can be computed with the maximum cardinality matching algorithm~\cite{Edmonds65}. We use an upper bound to save computational overhead.}. 
We assume the number of RF chains at the eNB satisfies $R \le L$, where $L$ is the number of mmBSs that are directly connected to the eNB, because $L$ RF chains is enough to serve the mmBSs. We have the following observation.
\begin{obs} 
If $k$ mmBS-to-mmBS links are active at a time $t$, then at most $\min(R, W - 2k)$ eNB-to-mmBS links can be active at $t$, each using one RF chain of the eNB. Hence, at least $R - \min(R, W - 2k) = \max(0, R-W + 2k)$ RF chains of the eNB are idle at $t$. 
\label{obs:idle-RF}
\end{obs}
 For a schedule of unit time, we define $t_k'$ as the time in which exactly $k$ mmBS-to-mmBS links are active. Each feasible schedule should be subject to the following constraints.
\begin{IEEEeqnarray}{rCl}
  \IEEEyesnumber\label{eq:ec-step1}\IEEEyessubnumber*
  \sum_{k=1}^{\nu}t_k'& \le &1,\label{eq:ec-12} 
  \textnormal{  and  }	t_k' \ge 0\quad\forall\,k\in\{1,2,...,\nu\}\label{eq:ec-11}
  \\
  \sum_{k=1}^{\nu}k\cdot t_k'& = &
  \sum_{e\in\{(v,v')\in\set{E}\colon v,v'\in\set{W}\}}t_e\label{eq:ec-13}
  \\
  \sum_{e\in\delta(\text{eNB})}t_e& \le &
  R-\sum_{k = 1}^{\nu}\max(0,R-W+2k)t_k'\label{eq:ec-14}
  \\  	
  \sum_{e\in\delta(v)} t_e& \le &1\quad\forall\,v \neq \text{eNB},
  \textnormal{  and  }
  % \sum_{e \in \delta(\text{eNB})} t_e &\le R \label{eq:ec-16}\\
  t_e  \ge 0\quad\forall\,e \in\set{E}.\label{eq:ec-15}
\end{IEEEeqnarray}
\eqref{eq:ec-13} formulates the total mmBS-to-mmBS transmission time in terms of the variables $t_k'$ and $t_e$ (link time of $e$), respectively. \eqref{eq:ec-14} shows that the total eNB-to-mmBS transmission time should be no more than $R$ minus the minimum idle time of the RF chains at the eNB. \eqref{eq:ec-15} expresses the single RF chain constraint on mmBSs. 

We substitute the precise constraint set to link times \eqref{eq:link-active-time} with the constraints in \eqref{eq:ec-step1}. The advantage is the low runtime and small memory complexity of linear programming due to the following reason.
The total number of constraints in \eqref{eq:link-active-time} and \eqref{eq:ec-step1} are $O(2^{W+R})$ (exponential) and $O(W^2)$ (polynomial), respectively.
Moreover, with the polynomial number of constraints, the computation of link time can be carried out by off-the-shelf linear optimization tools. However, the link time vector that satisfies \eqref{eq:ec-step1} may be infeasible in unit time, because satisfying the constraints in  \eqref{eq:ec-step1} is a necessary but not sufficient condition for a feasible unit time schedule.
%Given the same optimization problem, using the constraints in \eqref{eq:ec-step1} should give no worse a solution than using the constraints in \eqref{eq:link-active-time}, as the former is less tight than the latter. However, the former solution may be infeasible to schedule within unit time.

\begin{comment}
For the METR problem, the linear program for computing link time is
\begin{equation}
\begin{aligned}
	\min \quad &\sum_{l \in E} t_l \\
	\text{s.t.} \quad \sum_{l \in \delta^-(u)} b_l t_l - \sum_{l \in \delta^+(u)} b_l t_l &\ge h_u \quad \forall u \in \set{W} \\
	\text{and \eqref{eq:ec-step1}}.
\end{aligned}
\end{equation}
\end{comment}

Specifically, for the MTFS problem, the linear program for computing the max-min throughput $\theta$  is
\begin{IEEEeqnarray}{lcCrl}
\label{eq:ec-step1-theta}
\IEEEyesnumber \IEEEyessubnumber*
  \textnormal{maximize} & \theta & & & \\
  \textnormal{subject to}\quad&\sum_{e\in\delta^-(v)}c_e t_e
  -\sum_{e\in\delta^+(v)}c_e t_e& \ge &\theta
  \quad\forall\,v\in\set{W}\IEEEeqnarraynumspace\label{eq:ec-step1-theta1}
  \\
  &\text{and \eqref{eq:ec-step1}},& & &\nonumber
\end{IEEEeqnarray}
where $\delta^-(v)$ and $\delta^+(v)$ are the set of links coming into node $v$ and the set of links leaving $v$, respectively.

With the optimal $\theta$, we compute the link time for the MTFS.
\begin{IEEEeqnarray}{lll}
  \textnormal{maximize} &\sum_{e\in\delta^+(\text{eNB})}c_e t_e&
  \nonumber\\
  \textnormal{subject to}\quad&\textnormal{constraints in \eqref{eq:ec-step1-theta}.}&
  \IEEEeqnarraynumspace\label{eq:ec-step1-tput}
\end{IEEEeqnarray}

\begin{comment}
\begin{figure}
\centering
	\includegraphics[width=0.4\columnwidth]{figs/ec-step1-tight.eps}
\caption{A backhaul network. The numbers are link capacities. }
\label{fig:ec-step1-tight}	
\end{figure}

If we remove the constraints based on the idea of idle RF chains (Observation~\ref{obs:idle-RF}) from \eqref{eq:ec-step1}, and add the constraint of $R$ RF chains at eNB, we get the so-called {\em loose
  set} in the following~\eqref{eq:loose-set}, which may lead to
infeasible max-min throughput and network throughput.
\begin{IEEEeqnarray}{rCll}
  \sum_{e\in\delta(\text{eNB})}t_e& \le &R\nonumber\\
  \sum_{e\in\delta(v)}t_e& \le &1 &\qquad\forall\,v\in\set{W}
  \label{eq:loose-set}\\
  t_e&\ge &0 &\qquad\forall\,e\in\set{E}  \nonumber
\end{IEEEeqnarray}

As an example to show the effects of the constraints based Observation~\ref{obs:idle-RF}, we consider a simple network in Fig.~\ref{fig:ec-step1-tight}. Suppose that the eNB has $R = 2$ RF chains. For the given MTFS problem, both the optimal algorithm and the link time computation of the EC algorithm (\eqref{eq:ec-step1-theta} and \eqref{eq:ec-step1-tput}) achieves the same max-min and network throughput of $2.429$ and $7.286$ respectively. However, if we replace \eqref{eq:ec-step1} with the loose set \eqref{eq:loose-set} in link time computation, the max-min and network throughput are higher---$2.577$ and $7.731$, respectively. This evidences that \eqref{eq:ec-step1} is tighter than \eqref{eq:loose-set} in approximating the precise scheduling polyhedron \eqref{eq:link-active-time}.

\end{comment}

\subsection{Step (ii): edge-coloring based scheduling}
After we obtain the link time, the next step is to generate a unit time schedule. The approximation algorithm  is based on the idea of edge-coloring of multigraphs (graphs allowing multiple edges between two nodes). A proper edge-coloring assigns a color to each edge in a graph such that any two adjacent edges (sharing one or two common nodes) are assigned different colors. Obviously, the set of edges  $\set{E}_{\lambda}$ of a color $\lambda$ must be a matching. Hence, $\set{E}_\lambda$ corresponds to a slot and an edge coloring scheme corresponds to a schedule. The EC-based scheduling takes a parameter {\em granularity} $t^g \in (0, 1]$, which is the quantization of the link time. 
A smaller $t^g$ typically leads to better schedules at the cost of longer runtime. Alg.~\ref{alg:ec-sched} shows the process of the EC-based scheduling.
\begin{algorithm}[!t]
	\textbf{Reduce graph}. Given the link time vector $\vect{t}$, we remove edges in $G$ with zero link time and call the subgraph $G_r$\;
	\textbf{Expand eNB}. We perform the eNB expansion on $G_r$. Let the expanded eNBs be $\text{eNB}_1,...,\text{eNB}_R$.
For a given link $(\text{eNB}, v)$ in $G_r$ with link time $t_{(\text{eNB}, v)}$, we set the link time of the links $(\text{eNB}_k, v), k=1,...,R$ to $\frac{t_{(\text{eNB}, v)}}{R}$. We call the graph after eNB expansion $G_v = (\set{V}_v, \set{E}_v)$\;
	
	\textbf{Create multigraph and assign link time}. We create the
  coloring graph $G_m=(\set{V}_m,\set{E}_m)$, which has the same
  vertex set as $\set{V}_v$, and its edges is defined as follows. For each $e\in \set{E}_v$ between two nodes $v$ and $v'$ with link time $t_e$, we install $\bigl\lceil\frac{t_e}{t^g}\bigr\rceil$ edges
  between $v$ and $v'$ in $G_m$. Among these edges, $\bigl\lceil\frac{t_e}{t^g}\bigr\rceil-1$ edges are
  assigned $t^g$ link time, and the left edge is
  assigned $\bmod(t_e,t^g)$ link time ($\bmod$ is the modulo operation)\;
  
  \textbf{Coloring and scheduling}. We perform edge coloring on
  $G_m$. Suppose that $G_m$ can be edge-colored with $\kappa$
  colors. For those edges colored by the $i$-th color,
  $i=1,\ldots,\kappa$, we schedule the corresponding links in the $i$-th slot\ (a slot has the length $t^g$);

  \textbf{Scale}. The schedule is now of length $\kappa t^g$. If
  $\kappa t^g > 1$, we scale the total time length with the factor
  $\frac{1}{\kappa t^g}$\;

\caption{EC-based scheduling.}
\label{alg:ec-sched}
\end{algorithm}

\subsection{Performance analysis of the EC-based scheduling}

The following lemma shows that Alg.~\ref{alg:ec-sched} (step (ii) of the EC algorithm) reduces the performance metric $\mu~$\footnote{$\mu$ is the optimization goal such as throughput, energy consumption, etc.}  and link times $t_e$ by a factor of $\frac{1}{\kappa t^g}$, if the $\kappa t^g > 1$. Therefore, a high quality edge-coloring heuristic (small $\kappa$)~\cite{Nakano95} and a small $t^g$ improve the schedule performance.

\begin{lem} Suppose that after step (i) of the EC algorithm, each link $e$ of $G_v$ has link time $t_e$, and the performance metric
is $\mu \ge 0$. Moreover, assume that if a schedule is scaled by $\rho > 0$, then $\mu$ is also scaled by $\rho$. Therefore, after step (ii), the final link time is $t_e' = \min(\frac{t_e}{\kappa t^g}, t_e)$ and the final performance metric is $\mu' = \min(\frac{\mu}{\kappa t^g}, \mu)$.
\label{lem:ec-quality}
\end{lem}
\begin{proof}
	If $\kappa t^g \le 1$, then $G_v$ can be scheduled in unit time, and $t_e' = t_e$ and $\mu' = \mu$.
On the other hand, if $\kappa t^g > 1$, then $G_v$ needs $\kappa t^g$ time to schedule. To fit in the unit time schedule, we perform the scaling. Afterwards, $t_e' = \frac{t_e}{\kappa t^g}$ and $\mu' = \frac{\mu}{\kappa t^g}$.
\end{proof}

Since minimum edge coloring of an arbitrary graph is NP-complete~\cite{Holyer81}, we have to employ approximation algorithms. We choose a simple multigraph edge-coloring algorithm by Karloff et al~\cite{Karloff87}. It uses at most $3\left \lceil \Delta(G)/2 \right \rceil$ colors, where $\Delta(G)$ is the maximal node degree of a multigraph $G$. The following lemma gives the upper bounds on $\Delta(G_m)$, the number of vertices $|\set{V}_m|$ and edges $|\set{E}_m|$ of $G_m$.
\begin{lem}
$\Delta(G_m) \le W + R + {1} / {t^g} -1$. $|\set{V}_m| \le W + R$ and $|\set{E}_m| < \frac{1}{2} \left( W^2 + (2R-1)W + \frac{W + R}{t^g} \right)$.
\label{lem:ec-bounds}
\end{lem}
\begin{proof}
	$G_v$ contains $W + R$ nodes due to the eNB expansion. For a expanded eNB node, the maximum degree is no more than $W$. For a mmBS, the maximum degree is no more than $W-1 + R$ since the directed graph $G_v$ contains no cycles as all cycles can be eliminated by shortening the link time. So $\Delta(G_v) \le W + R - 1$. $G_m$ is transformed from $G_v$ by installing $\bigl\lceil\frac{t_e}{t^g}\bigr\rceil$ edges for each edge $e$ in $G_v$. Therefore, the degree of a node $v$ in $G_m$ is 
\begin{equation*}
	\deg(v) = \sum_{\substack{e \in \delta(v) \\ v \in \set{V}_v}} \left\lceil \frac{t_e}{t^g} \right\rceil
	< \sum_{\substack{e \in \delta(v) \\ v \in \set{V}_v}} \left( \frac{t_e}{t^g} + 1 \right) \le \frac{1}{t^g} + W + R - 1 
\end{equation*}
In addition, we have
$
	|\set{V}_m| = |\set{V}_v| \le W + R
$.
Using the \emph{degree sum formula}, the number of edges in $G_m$ is
\begin{IEEEeqnarray*}{rCl}
  \card{\set{E}_m}& = &\frac{1}{2}\sum_{v\in\set{V}_m}\deg(v)
  <\frac{1}{2}\sum_{v\in\set{V}_v}\sum_{e\in\delta(v)}\left(\frac{t_e}{t^g}+1\right)
  \\
  & \le &\frac{1}{2}\left(W^2+(2R-1)W+\frac{W+R}{t^g}\right). \qedhere
\end{IEEEeqnarray*}
\end{proof}

Since $R \le W$ ($W$ RF chains is sufficient to serve all mmBSs), from the above Lemma, we have $\Delta(G_m) = O(W + \frac{1}{t^g})$. $|\set{V}_m| = O(W)$ and $|\set{E}_m| = O(W^2+\frac{W}{t^g})$.

In the following, we show the quality and time complexity of the step (ii) of the EC algorithm. Practically, step (i) is always much faster than step (ii).
\begin{thm}
	Let the performance metric after step (i) be $\mu \ge 0$. Then step (ii) achieves the performance metric $\mu' > \frac{2}{3[(W+R+1)t^g + 1]} \mu$ and it has time complexity of $O(\bigl[W^2 + \frac{W}{t^g}\bigr]\log(W+\frac{1}{t^g}))$.
\label{thm:ec-perf}
\end{thm}
\begin{proof}
	The Karloff's algorithm uses $\kappa \le 3 \left\lceil \Delta(G_m) /{2} \right \rceil$ colors. Due to Lemma~\ref{lem:ec-bounds}, 
$
	\kappa \le 3 \left \lceil {(W+R+\frac{1}{t^g}-1)} / {2} \right \rceil.
$
Hence,
\begin{IEEEeqnarray*}{rCl}
  \kappa t^g& < &3t^g\biggl(\frac{W+R+\frac{1}{t^g}-1}{2}+1\biggr)
  =\frac{3}{2}(W + R + 1)t^g + \frac{3}{2}.
\end{IEEEeqnarray*} 
According to Lemma~\ref{lem:ec-quality}, the final performance metric
\begin{equation*}
	\mu' = \min(\frac{1}{\kappa t^g}, 1) \mu 
	> \frac{2}{3[(W+R+1)t^g + 1]} \mu
\end{equation*}
The time complexity of step (ii) is determined by Karloff's edge coloring algorithm,
which has the same time complexity as the perfect edge coloring of a bipartite graph of maximum degree $\lceil \Delta(G_m)/2 \rceil$ and number of edges $O(|\set{E}_m|)$. Since the time complexity of the perfect edge coloring of a graph with $|\set{E}|$ edges and degree $\Delta$ is $O(|\set{E}| \log \Delta)$~\cite{Cole01},
the time complexity of step (ii) is $O(|\set{E}_m| \log (\lceil \Delta(G_m)/2 \rceil)) = O(\bigl[\frac{W}{t^g}+W^2\bigr]\log(W+\frac{1}{t^g}))$.
\end{proof}
For the MTFS problem, let the optimal max-min throughput be $\theta^*$. Since  step (i) of the EC algorithm uses a looser constraint set than the precise set of the schedule polyhedron, it gives a $\theta \ge \theta^*$. From Theorem~\ref{thm:ec-perf}, we have that the final max-min throughput of the EC algorithm $\theta' > \frac{2}{3[(W+R+1)t^g + 1]} \theta^*$.

In a typical backhaul network, we have $W \gg R$, so $\frac{2}{3[(W+R+1)t^g + 1]} \approx \frac{2}{3(Wt^g + 1)}$. This means to keep a constant performance quality, we can choose $t^g$ to be inversely proportional to $W$, i.e., a bigger network requires a smaller $t^g$. If $t^g$ is so selected, then the time complexity is $O(W^2 \log (W)$, which is quite scalable with the number of mmBSs, and thus feasible at the eNB in practice.
Moreover, by setting the granularity $t^g \rightarrow 0$, the performance metric approaches $\mu' > \frac{2}{3} \mu$. For the MTFS problem, this means $\theta' > \frac{2}{3} \theta^*$.

\section{Extension to more general scenarios} \label{s:extend}
In this section, we show that our schedule-oriented optimization method proposed in Sec.~\ref{s:sched-fair} and the edge-coloring based approximation algorithm in Sec.~\ref{s:approx-algo} can be extended to more general scenarios of (i) backhaul and access networks and (ii) multiple RF chains at each node.

\subsection{Extension to backhaul and access networks}
The backhaul and access networks add an additional layer of UEs to the backhaul networks. Each UE has a single RF chain and is allowed to have links with one or more mmBSs.
Let $\set{U}$ denote the set of UEs. For the downlink traffic, only one-directional links from mmBSs to UEs exist.

\begin{comment}
Fig.~\ref{fig:sys_model_ba} shows a toy example of backhaul and access network. 

\begin{figure}[bpth]
\centering
	\includegraphics[width=0.5\columnwidth]{figs/sys_model_ba.eps}
\caption{A backhaul and access network example. The edges shown are the potential links for downlink schedule. A bidirectional edge represents two opposite links. A UE is allowed to be associated with multiple mmBSs. The network is based on the backhaul network in Fig.~\ref{fig:sys_model}.}
\label{fig:sys_model_ba}	
\end{figure}
\end{comment}

We illustrate as an example the solution of the MTFS problem.
The max-min fairness in throughput is now defined for the UEs, as they are the destinations.
Let $G$ be the backhaul and access network after eNB expansion, and $\mat{A}$ be the node-matching matrix of $G$. We define $\mat{A}^{\set{W}}$ and $\mat{A}^{\set{U}}$ as the submatrices of $\mat{A}$ related to the nodes in $\set{W}$ (mmBSs) and in $\set{U}$ (UEs), respectively. The linear program \eqref{eq:mtf-theta} to compute max-min throughput $\theta$ needs to be modified as follows: \eqref{eq:mtf-theta1} should be replaced by the constraints of \eqref{eq:mtf-theta-ba}.
\begin{IEEEeqnarray*}{rCl}
  \IEEEyesnumber\label{eq:mtf-theta-ba}
	\mat{A}^{\set{U}}\vect{t}^\schd{S}& \ge &\theta\vect{1}
  \textnormal{ and }
  \mat{A}^{\set{M}}\vect{t}^\schd{S} = \vect{0}
\end{IEEEeqnarray*}
here, \eqref{eq:mtf-theta-ba} expresses the constraint that the throughput at a UE must be at least $\theta$ and each mmBS is a pure relay. Obviously, the new MTFS problem can be solved with the same optimization technique as proposed in Sec.~\ref{s:sched-fair}. %The only caveat is that practically large quantities of UEs exist in the network and the mmBS-to-UE links are more dynamic than the mmBS-to-mmBS ones due to the mobility of end users, which may result in high computational overhead.

As for the EC algorithm, we need to replace the data flow constraint of \eqref{eq:ec-step1-theta1} with the following \eqref{eq:ec-step1-theta-ba}.
\begin{IEEEeqnarray}{lCrl}
\IEEEyesnumber\label{eq:ec-step1-theta-ba}\IEEEyessubnumber*
\sum_{e \in \delta^-(v)} c_e t_e - \sum_{e \in \delta^+(v)} c_e t_e& = &0 \quad &\forall \, v \in \set{W} \\
\sum_{e \in \delta^-(v)} c_e t_e - \sum_{e \in \delta^+(v)} c_e t_e& \ge &\theta \quad &\forall \,v \in \set{U}
\end{IEEEeqnarray}

\subsection{Extension to multiple RF chains at each node}
Now we remove the restriction that all nodes except the eNB have single RF chain. The technique to deal with this problem is the so-called {\em node expansion} which extends the eNB expansion. 

\textbf{Node expansion.} Let $G$ be a directed graph representing the network. 
We create a expanded graph $G'$ as follows: for each node $v$ in $G$, we create $R_v$ {\em expanded nodes} $v_1, ..., v_{R_v}$ in $G'$ where $R_v$ is the number of RF chains at $v$. The expanded nodes of $v$ are collectively called a {\em super node} $v$ in $G'$.
Moreover, if there is a link of capacity $c$ between node $v$ and node $v'$ in $G$, then we install $R_v \cdot R_{v'}$ links between all combinations of $v_i, i=1,...,R_v$ and $v'_j, j=1,...,R_{v'}$. Each link $(v_i, v'_j)$ is assigned the capacity $c$. An example is shown in Fig.~\ref{fig:node_expand}.

\begin{figure}[!t]
\centering
	\includegraphics[width=0.5\textwidth]{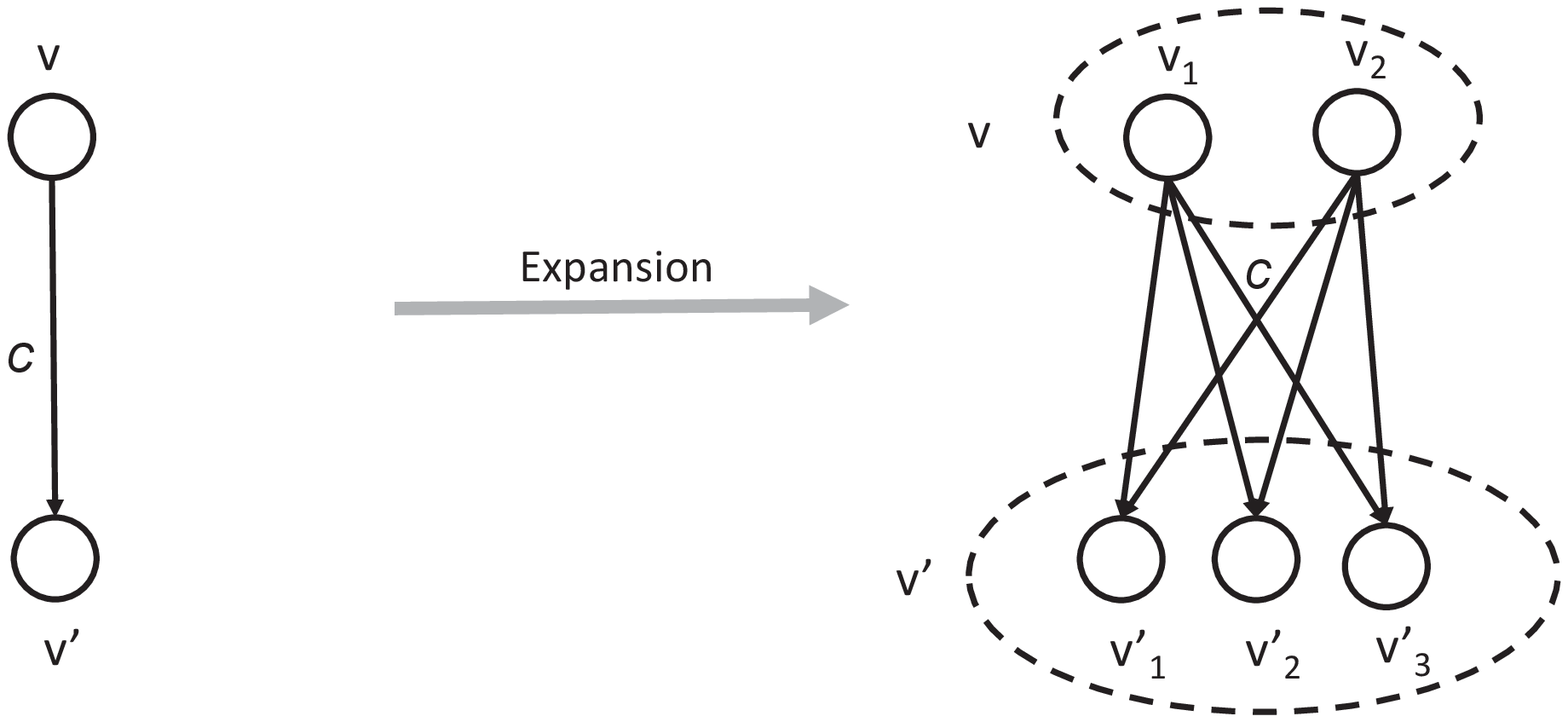}
\caption{Node expansion. $v$ has 2 RF chains and $v'$ has 3 RF chains.}
\label{fig:node_expand}	
\end{figure}

After the node expansion, the constraint of the RF chains is implicitly guaranteed by the matchings in $G'$. For the MTFS problem, the definition of the node-matching matrix $\mat{A}^\set{W}$ needs adaptation because now multiple links can be incident to a super node in a matching in $G'$. A row of $\mat{A}^\set{W}$ corresponds to an mmBS super node  (a collection of expanded nodes) and a column of $\mat{A}^\set{W}$ corresponds to a matching in $G'$. Therefore, an element of the matrix $a^\set{W}_{i,j}$ has the value of the sum capacity of all links entering the $i$-th mmBS super node in the $j$-th matching minus the sum capacity of all links leaving the $i$-th mmBS super node in the $j$-th matching.

As for the EC algorithm, we need to modify the necessary schedule constraints in \eqref{eq:ec-step1}. Since the maximum number of simultaneous mmBS-to-mmBS links increases for $R$ times when each mmBS has $R$ RF chains. This leads to $R$ times more variables of $t'_k$ in \eqref{eq:ec-step1}, which may lead to long computation time in linear program. So, we delete the constraints on $t'_k$. The following simple set of necessary schedule constraints is used to replace \eqref{eq:ec-step1},
\begin{IEEEeqnarray*}{rCl}
\sum_{e \in \delta(v)} t_e &\le& R_v\quad\forall\,v\in \set{V},
\textnormal{ and }t_e\ge 0 \quad\forall \, e \in \set{E}.
\end{IEEEeqnarray*}

Another modification is to replace the {\em expand eNB} step in Alg.~\ref{alg:ec-sched} with the following step.

\noindent \textbf{Expand node}. We perform the node expansion on $G_r$. Let the resulting graph be $G_v$. 
For a given link $(v, v')$  in $G_r$ with link time $t_{(v, v')}$, we assign the link time $\frac{t_{(v, v')}}{R_v R_{v'}}$ to the links $(v_i, v'_j)$ in $G_v$ for all combinations of $i = 1,...,R_v$ and $j = 1,...,R_{v'}$.

%From then on, the proposed optimization method can be carried out. 

\section{Numerical Evaluation} \label{s:eval}
In this section, we evaluate the optimal MTFS algorithm and EC-based approximation algorithm in terms of max-min throughput, network throughput and runtime efficiency.

\subsection{Evaluation setting}
We simulate a mmWave backhaul network, where $n \times n$ mmBSs are placed on the intersections of a $n \times n$ grid and the eNB is placed in the center of the grid. The distance between two neighboring mmBSs is $d_g$. The capacity of each link is calculated with the channel model described in Sec.~\ref{s:sysmodel}. We assume a carrier frequency of 28 GHz. The channel state between any two nodes is simulated according to the statistical model derived from the real-world measurement~\cite{Akdeniz14}. The channel state has three possibilities---LOS (line-of sight), NLOS (non line-of-sight) or outage.
 The simulation parameters are listed in Tab.~\ref{tb:sim-param}.

\begin{table}[t!]
\setlength\extrarowheight{2pt}
\caption{Simulation parameters}
\centering
\footnotesize
\label{tb:sim-param}
\begin{tabular}{|p{3.9cm}| p{4.2cm} |}
\hline
\textbf{ Parameter }	 &	\textbf{ Value } 		\\
\hline
\hline
Distance between 2 mmBSs, $d_g$ 				& 100 m	\\ \hline
%Number of RF chains at eNB, $R$ & 10 \\ \hline
\multirow{2}{*}{\shortstack[l]
{Path loss parameters $\alpha, \beta, \sigma$ in \\ $PL ( d )  = \alpha + 10 \beta \log _ { 10 } { d } + \xi$}} 	& LOS: $\alpha = 61.4, \beta = 2, \sigma = 5.8$ \\ 
& NLOS: $\alpha = 72, \beta = 2.92, \sigma = 8.7$ \\ \hline
Transmission power, $p_{\text{tx}}$						& 30 dB \\ \hline
Directivity gain, $g_\text{x}$		& 30 dB \\ \hline
Bandwidth, $b$		& 1 GHz \\ \hline
Noise $N_0 = kT_0 +F + 10 \log_{10} b$ & $kT_0 = -174$ dBm/Hz, $F = 4$ dB \\ \hline
Minimum SINR threshold, $\tau$	& $-5$ dB \\ \hline
\end{tabular}
\end{table}

The proposed algorithms are implemented in MATLAB, with the exception that the optimal MTFS algorithm uses a C++ implementation for maximum weighted matching~\cite{Kolmogorov09}. 

%The software environment for the simulation is MATLAB R2017a running on Ubuntu 16.04. The hardware environment is a computer equipped with a 2.6 GHz Intel Core i7-6600U CPU and 16 GB of RAM.

\subsection{mmBS with single RF chain}

We evaluate the MTFS scheduling by varying the number of mmBSs from $4 \times 4$ to $16 \times 16$. The eNB has $R = 10$ RF chains and all the other nodes have single RF chain.
For each network size, 30 instances of link capacities are randomly generated and then the network is scheduled for the MTFS problem. The performance results are shown in Fig.~\ref{fig:backhaul_perf}. As expected, the optimal MTFS algorithm (OPT-MTFS) always attains the highest max-min throughput (Fig.~\ref{fig:backhaul_perf}(a)). In contrast, the max-min throughput of the EC-based approximation algorithm (EC) is smaller. However, the value is significantly better than the theoretical lower bound of Theorem~\ref{thm:ec-perf} (dashed lines in the figure).
It goes up with an increase in granularity (corresponding to a lower $t^g$). Practically, $t^g$ in the range of 0.01 to 0.001 is ideal for the network size of up to 200 nodes, as on average, the EC algorithm achieves $70\%$ to $90\%$ of the optimal max-min throughput. The price for the high max-min throughput is a decrease in runtime efficiency. For a network of 256 mmBSs, the EC algorithm with $t^g = 0.001$ runs 100x faster than the optimal MTFS algorithm (Fig.~\ref{fig:backhaul_perf}(b), note the log scale). Moreover, with the increase of the number of nodes, the runtime of the EC algorithm grows more slowly than the optimal MTFS algorithm, which shows the better scalability of the former for large networks. 

As the goal of the MTFS is to maximize the network throughput under the fairness condition, we also compare the network throughput of the optimal MTFS algorithm, the EC algorithm and the unconditional maximum network throughput (MAX-TPUT). The MAX-TPUT achieves the maximum network throughput for a given network.
It is obtained when the $\min(R, L)$ ($L$ is the number of mmBSs directly connected to the eNB) eNB-to-mmBS links with the highest capacities are active throughout the unit schedule and all the other links are inactive. If the number of mmBSs $W > L$, some of the mmBSs will have zero throughput. This is the worst case with respect to max-min fairness in throughput.
 Our evaluation results show that, on average, max-min fairness limits the network throughput to be approximately half of the maximum value. Since the simulated backhaul network is well-connected, the network throughput is in most cases equal to the max-min throughput times the number of mmBSs. 
Therefore, the relative performance of the network throughput between the optimal MTFS algorithm and the EC algorithm at different granularities is almost the same as that of the max-min throughput in Fig.~\ref{fig:backhaul_perf}(a).

%Occasionally, the EC algorithm gets slightly higher throughput than that of the optimal algorithm, since the network throughput may be higher at a lower max-min throughput. 

%We have shown in Sec.~\ref{s:approx-algo} that the EC algorithm replaces the precise but exponential number of constraints for link time in \eqref{eq:link-active-time} with the polynomial number of necessary constraints in \eqref{eq:ec-step1}. The simulation shows that for the evaluated backhaul networks, \eqref{eq:ec-step1} gives max-min throughput (optimization goal) very close to the optimal value. We also compare it with the looser constraints 

\begin{figure}[!t]
\centering
\subfigure[max-min throughput $\theta$]{\includegraphics[width = 8.8cm]{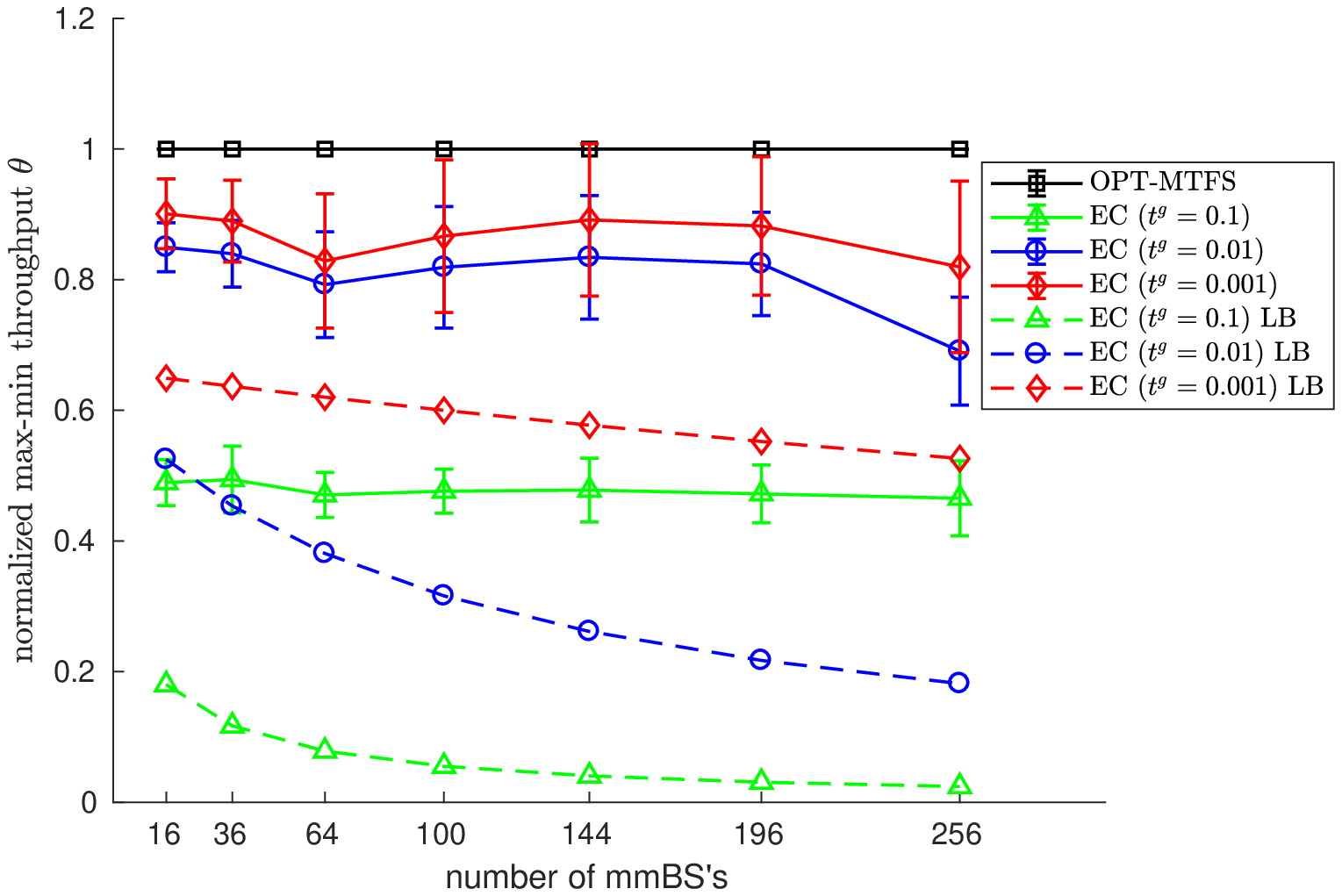}}
\subfigure[run time]{\includegraphics[width = 8cm]{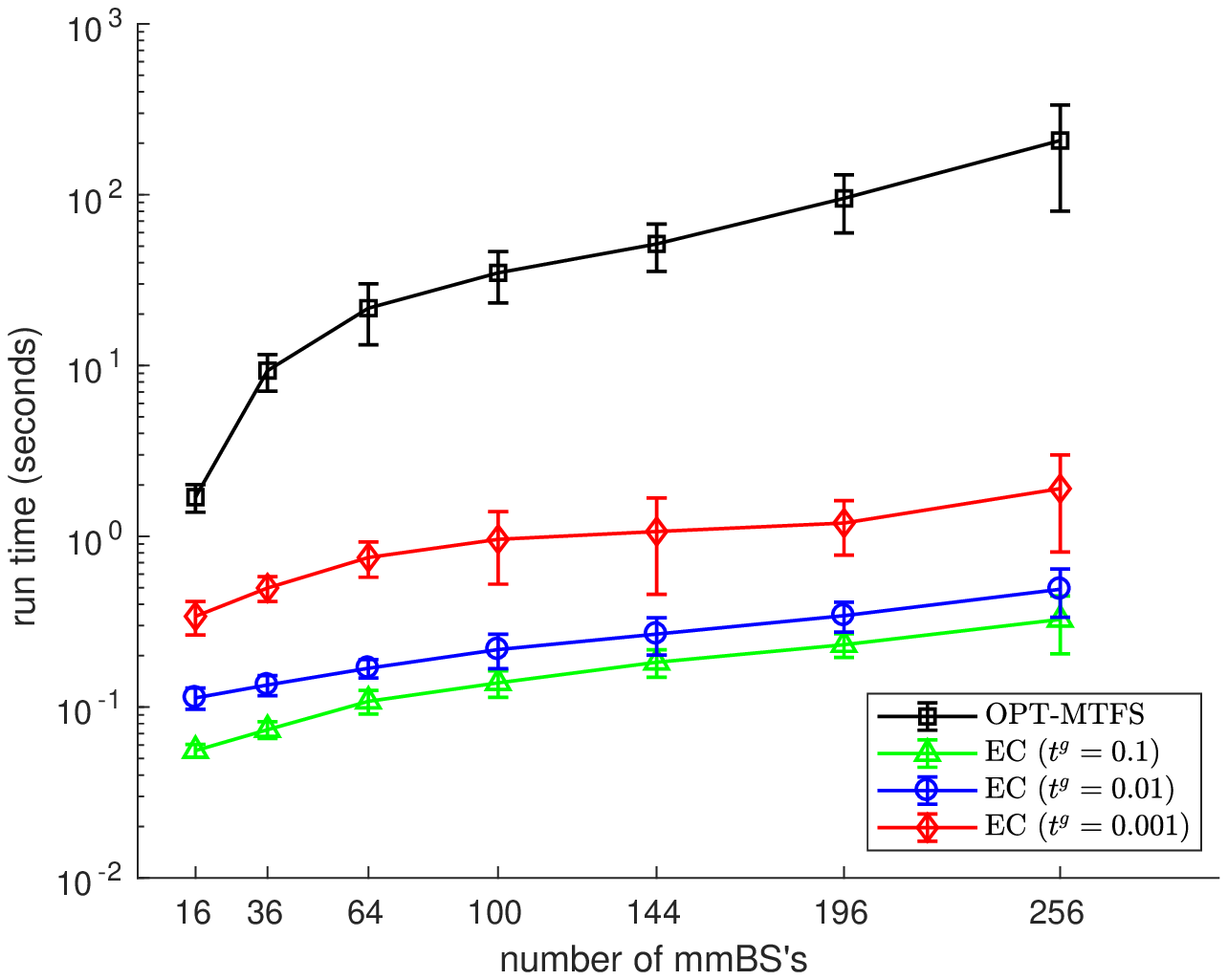}} 
%\subfigure[network throughput]{\includegraphics[width = 8cm]{figs/norm-tput.eps}}
\caption{Performance of the optimal MTFS algorithm (OPT-MTFS) and the EC-based approximation algorithm (EC) under the condition of single-RF-chain mmBSs.
The curves show the average performance, and the error bars show $\pm$ standard deviation.}
\label{fig:backhaul_perf}
\end{figure}

\subsection{mmBS with multiple RF chains}
We now evaluate the performance for the situation that each mmBS is equipped with multiple RF chains. For that purpose, we simulate a backhaul network with $10 \times 10$ mmBSs and evaluate the cases that the eNB has 10 RF chains and each mmBS has $R_{\set{W}}$ RF chains, with $R_{\set{W}}$ varying from 1 to 10. For each given $R_{\set{W}}$, 30 instances of random link capacities are generated. 

\begin{figure}[!t]
\centering
\subfigure[max-min throughput $\theta$]{\includegraphics[width = 8cm]{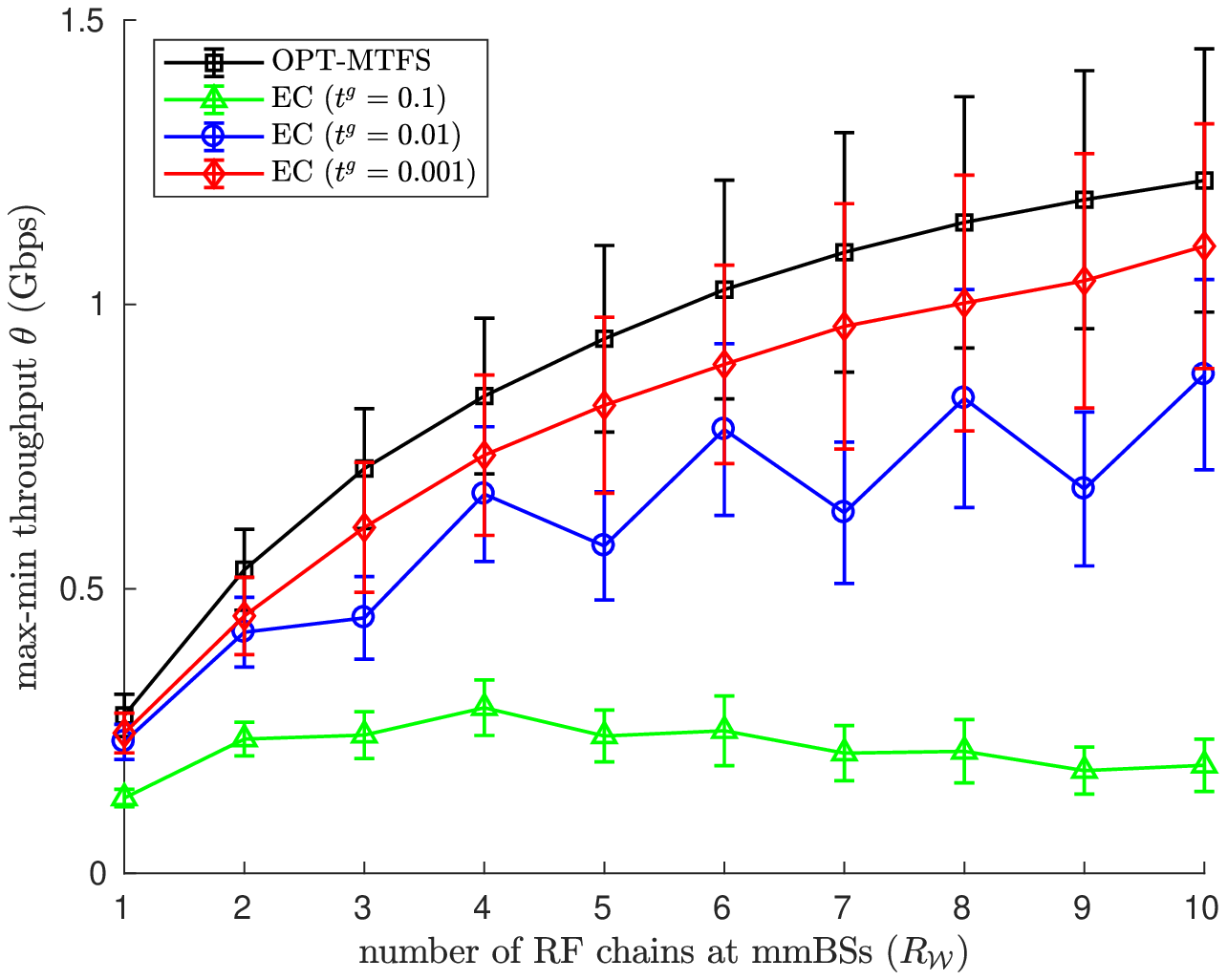}}
\subfigure[run time]{\includegraphics[width = 8cm]{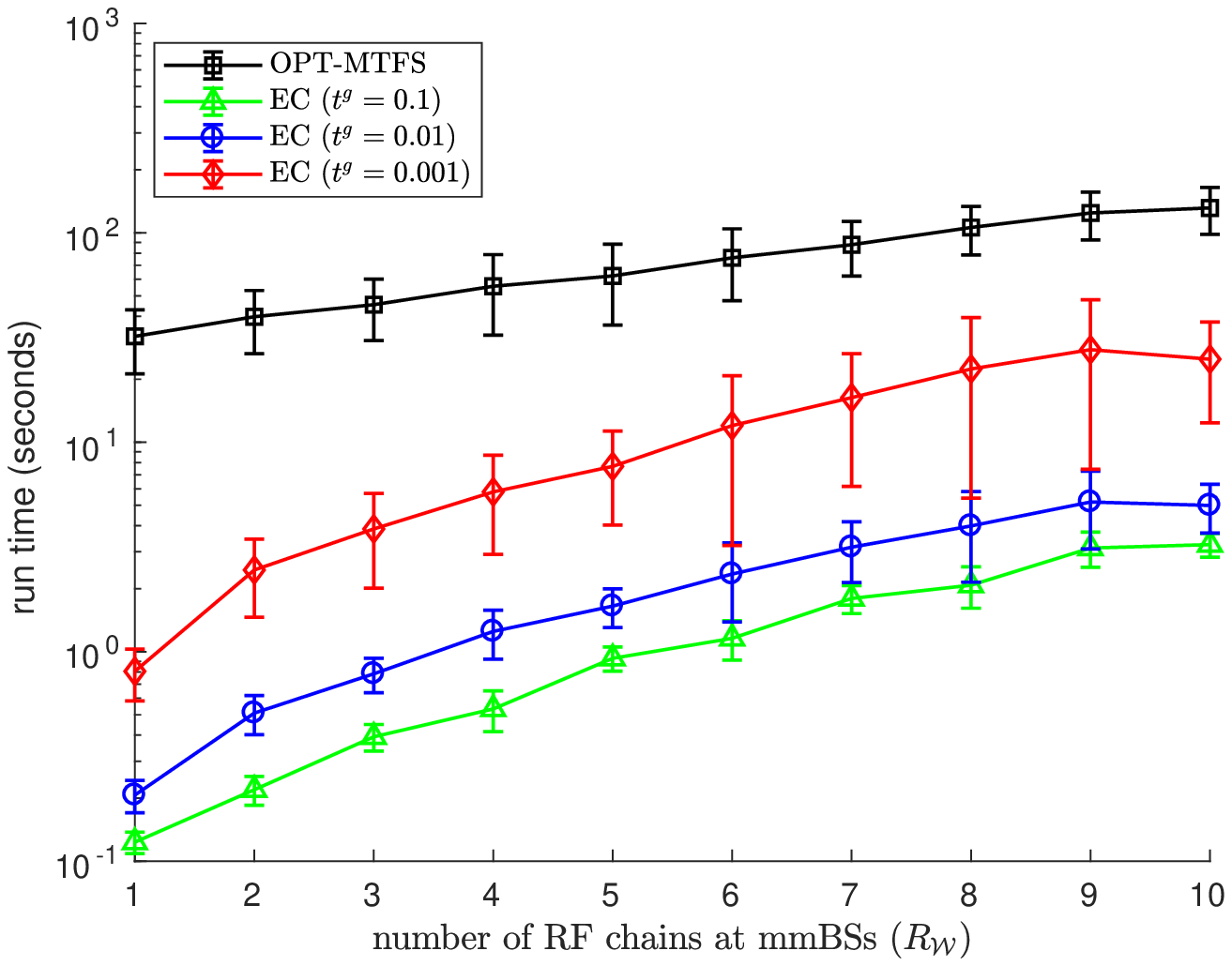}} 
\caption{Performance comparison for different number of RF chains at mmBSs. The curves show the average performance, and the error bars show $\pm$ standard deviation.}
\label{fig:comp-rf}
\end{figure}

As shown in Fig.~\ref{fig:comp-rf}(a), the max-min throughput goes up steadily with $R_{\set{W}}$. The optimal max-min throughput at $R_{\set{W}} = 10$ is over 4 times higher than the value at $R_{\set{W}}=1$. 
The difference in performance is due to the larger number of simultaneous links in the setting of multi-RF-chain mmBSs. Evaluation results show that the number of simultaneous links is almost proportional to $R_{\set{W}}$, as node expansion has increased the number of nodes for $R_{\set{W}}$ times.
Therefore, to attain higher throughput at each mmBS, an option is to equip mmBSs with multiple RF chains. However, the run time of the optimal MTFS algorithm also increases with $R_{\set{W}}$ (Fig.~\ref{fig:comp-rf}(b)). The extra time is spent in the maximum weighted matching in an expanded network with roughly $R_{\set{W}}$ times more nodes and $R_{\set{W}}^2$ times more edges. By using the EC-algorithm with $t^g = 0.001$, we achieve $85\%$ to $90\%$ of the optimal max-min throughput while using 2\% to 20\%  of the time.

\section{Conclusions}
\label{s:conclusions} 
The paper presents an optimal joint routing and scheduling method---{\em schedule-oriented optimization} for mmWave cellular networks based on matching theory. 
It can solve any problem that can be formulated as a linear program whose variables are link times and QoS metrics.
The method is demonstrated to be efficient in practice, capable of solving the maximum throughput fair scheduling (MTFS) problem within a few minutes for over 200 mmBSs. For better runtime efficiency, an edge-coloring based approximation algorithm is presented, which runs 5 to 100 times faster than the optimal algorithm while achieving over $80\%$ of the optimal performance. In summary, the proposed optimal and approximation algorithms are highly practical for mmWave cellular networks.

\section*{Acknowledgement}
This work has been supported by the German Research Foundation (DFG) in the Collaborative Research Center (SFB) 1053 ``MAKI: Multi-Mechanism-Adaptation for the Future Internet'' and by LOEWE NICER.
It has also been partially supported by the Minister of Science and Technology (MOST) of Taiwan under Grants MOST 103-2911-I-011-515 and MOST 104-2911-I-011-503. 
Lin's work was completed during his visit to the Center for Advanced Security Research Darmstadt (CASED), Technische Universit\"{a}t Darmstadt, Germany, during 2014 to 2016.

\ifthenelse{\boolean{long}}{\appendix
\subsection{Proof of Lemma~\ref{lm:schedule-matching}}
\label{sec:proof_lemma1}
\begin{proof} We first prove part (2). Let a feasible schedule $\schd{S}$ consist of $N \ge 1$ slots. Because the total length of all slots is no more than 1, we have
\begin{equation}
	\sum_{i = 1}^N t_{i} \le 1
\label{eq:slot_sum}
\end{equation}
In~\cite{Edmonds65b}, Edmonds' matching polyhedron theorem states that
  all matchings in a graph $G$ are one-to-one mapped to the vertices of
  the matching polyhedron $Q$ described by \eqref{eq:matching} (each
  vertex of $Q$ has elements of either $0$ or $1$, where $x_e = 1$
  means that the edge $e$ is in the matching).
\begin{IEEEeqnarray}{rCll}
    \IEEEyesnumber\label{eq:matching}\IEEEyessubnumber*
    \sum_{e \in \delta(v)} x_e & \le & 1 &\qquad \forall\,v \in
    \set{V},
    \\
    \sum_{e\in\set{E}(\set{O})} x_e& \le &
    \left\lfloor\frac{\card{\set{O}}}{2}\right\rfloor &\qquad
    \forall\,\text{odd set }\set{O}\subseteq\set{V}, \\
    x_e & \ge &0 &\qquad \forall\,e\in\set{E}.
  \end{IEEEeqnarray}
Since in each slot, the set of scheduled links is a matching in $G$,
  we further define $x_e^i=1$ if the link $e$ is active in $i$-th
  slot; otherwise, $x_e^i=0$. Then for a given index $i$, each
  variable $x_e^i, e \in \set{E}$ satisfies \eqref{eq:matching} when
  $x_e$ is replaced with $x_e^i$. Combining \eqref{eq:slot_sum} and
  \eqref{eq:matching}, we have
\begin{IEEEeqnarray}{rCl}
    \IEEEyesnumber\label{eq:matching-times}\IEEEyessubnumber*
    \sum_{i=1}^N t_{i}\sum_{e\in\delta(v)} x_e^i& \le &
    \sum_{i=1}^N t_{i},
    \label{eq:matching-times-1}\\
    \sum_{i=1}^N t_{i} \sum_{i\in\set{E}(\set{O})} x_e^i& \le &
    \sum_{i=1}^N \left\lfloor\frac{\card{\set{O}}}{2} \right\rfloor
    t_i.\label{eq:matching-times-2}    
  \end{IEEEeqnarray}
Since the link time $t_e = \sum_{i=1}^N t_i x_e^i\geq 0$,
  this implies that \eqref{eq:matching-times-1} and
  \eqref{eq:matching-times-2} are equivalent to
  \eqref{eq:node-matching} and \eqref{eq:oddset}, respectively. Thus,
  each feasible link time vector is a point in $P$.

We now prove part (1). Since the
  schedule polyhedron $P$ is the same as the matching polyhedron $Q$,
  each vertex of $P$ is a matching in $G$. Suppose that all the
  vertices of $P$ are $\vect{x}_1, ..., \vect{x}_K$, where $K$ is some
  positive integer. Since $P$ is a convex set, it means by definition
  that each point $\vect{t}\in P$ can be expressed by a convex
  combination of the vertices of $P$:
  \begin{IEEEeqnarray*}{rCl}
    \vect{t}=\sum_{k = 1}^K\alpha_k
    \vect{x}_k,
  \end{IEEEeqnarray*}
  where $\alpha_k\ge 0$ and $\sum_{k=1}^K\alpha_k=1$. This can be
  interpreted as follows: a point $\vect{t} \in P$ corresponds to a
  feasible schedule $\schd{S}$ of unit length. $\schd{S}$ has $K$
  slots and the length of the $k$-th slot is $\alpha_k$. Here, the
  links correspond to $\vect{x}_k$ is a matching, and they are
  scheduled in the $k$-th slot. Thus, we have proved that any
  $\vect{t}\in P$ is feasible.
\end{proof}

\subsection{Proof of Theorem~\ref{thm:MTFS-polynomial-time}}
\label{sec:proof_MTFS-polynomial-time}
The proof applies the technique used in~\cite{Nemhauser91} to prove
that fractional edge coloring can be solved in polynomial time by the
ellipsoid algorithm. Specifically, a linear program is solvable in
polynomial time if the separation problem of its dual problem can be
solved in polynomial time. The separation problem of a linear program
$J$ is to determine whether a given solution satisfies all constraints
of $J$ or a violated constraint is identified.  If we can solve both
linear programs of \eqref{eq:mtf-theta} and \eqref{eq:mtf} in
polynomial time, then we can solve the MTFS problem in polynomial
time.

We first prove that \eqref{eq:mtf-theta} can be solved in polynomial
time. The dual of \eqref{eq:mtf-theta} is
\begin{IEEEeqnarray}{lrCl}
  \IEEEyesnumber\label{eq:mtf-theta-dual}\IEEEyessubnumber*
  \textnormal{maximize} & q & &
  \\
  \text{subject to}\quad&
  \trans{\vect{p}}\mat{A}^\set{W}+q\trans{\vect{1}}& \le &\trans{\vect{0}}\label{eq:mtf-theta-dual-1}
  \\
  & \trans{\vect{p}}\vect{1}& = &1\label{eq:mtf-theta-dual-2}
  \\
  & \vect{p}& \ge &\vect{0}\label{eq:mtf-theta-dual-3}.
\end{IEEEeqnarray}
Given a solution $(\vect{p}, q)$, \eqref{eq:mtf-theta-dual-2} and
\eqref{eq:mtf-theta-dual-3} can be checked in polynomial time, since
the total number of constraints in \eqref{eq:mtf-theta-dual-2} and
\eqref{eq:mtf-theta-dual-3} is $W+1$ and $\vect{p}$ contains $W$
elements.

To check \eqref{eq:mtf-theta-dual-1}, we use the polynomial maximum
weighted matching algorithm~\cite{Edmonds65b}. A constraint in
\eqref{eq:mtf-theta-dual-1} is of the form
$\trans{\vect{p}}\vect{a}^\set{W}_k\le -q$, where $\vect{a}^\set{W}_k$
is the $k$-th column of $\mat{A}^\set{W}$ ($\vect{a}^{\set{W}}_k$
corresponds to a matching). We set the weights $w_{(v_i,v_j)}$ to the
links $(v_i,v_j)$ such that
\begin{IEEEeqnarray}{rCl}
  w_{(v_i,v_j)}=
  \begin{cases}
    c_{(v_i,v_j)} (p_j - p_i) & \text{if }v_i, v_j \in \set{W}
    \\
    c_{(v_i,v_j)} p_j         & \text{otherwise }v_i\in\set{R},v_j\in\set{W}.
  \end{cases}\IEEEeqnarraynumspace\label{eq:ftr-setw}
\end{IEEEeqnarray}

Then we perform maximum weighted matching on $G$. If the weight of the maximum
weighted matching satisfies $w \le -q$, then $(\vect{p}, q)$ satisfies
\eqref{eq:mtf-theta-dual-1}. Otherwise the maximum weighted matching
gives a violated constraint.

According to Theorem.~3.10 in~\cite{Groetschel81}, for a linear program
$J$, if we can solve the separation problem of its dual $J^*$ in
polynomial time, then we can solve both $J$ and $J^*$ in polynomial
time with the ellipsoid algorithm. This proves that
\eqref{eq:mtf-theta} can be solved in polynomial time.

Similarly, we next prove that \eqref{eq:mtf} can be solved in
polynomial time. The dual of \eqref{eq:mtf} is
\begin{IEEEeqnarray}{lrCl}
  \IEEEyesnumber\label{eq:mtf-dual}\IEEEyessubnumber*
  \textnormal{maximize} & \theta\trans{\vect{p}}\vect{1}+ q & &
  \\
  \text{subject to}\quad&
  \trans{\vect{p}}\mat{A}^\set{W}+q\trans{\vect{1}}
  & \le &-\trans{\vect{c}}\label{eq:mtf-dual-1}
  \\
  & \vect{p}& \ge &\vect{0}.
\end{IEEEeqnarray}
Given a tuple $(\vect{p}, q)$, we set the weights $w_{(v_i,v_j)}$ to
the links $(v_i,v_j)$ such that
\begin{IEEEeqnarray}{rCl}
  w_{(v_i, v_j)} =
  \begin{cases}
    c_{(v_i,v_j)} (p_j - p_i) & \text{if }v_i,v_j\in\set{W} \\
    c_{(v_i,v_j)} (p_j + 1)   & \text{otherwise }v_i\in\set{R},v_j\in\set{W}.
  \end{cases}\nonumber\\*\label{eq:ftr-setw2}
\end{IEEEeqnarray}
Then we perform maximum weighted matching on $G$. Depending on whether
the weight of the maximum weighted matching satisfies $w \le -q$, the constraints of
\eqref{eq:mtf-dual-1} are satisfied or a violated one is
identified. With the same argument as above, \eqref{eq:mtf} can be
solved in polynomial time. This complete the proof that the MTFS
problem can be solved in polynomial time with the ellipsoid method.
}{}

\bibliographystyle{IEEEtran}

%\scriptsize
%\footnotesize
\bibliography{biblio}

\begin{thebibliography}{10}
\providecommand{\url}[1]{#1}
\csname url@samestyle\endcsname
\providecommand{\newblock}{\relax}
\providecommand{\bibinfo}[2]{#2}
\providecommand{\BIBentrySTDinterwordspacing}{\spaceskip=0pt\relax}
\providecommand{\BIBentryALTinterwordstretchfactor}{4}
\providecommand{\BIBentryALTinterwordspacing}{\spaceskip=\fontdimen2\font plus
\BIBentryALTinterwordstretchfactor\fontdimen3\font minus
  \fontdimen4\font\relax}
\providecommand{\BIBforeignlanguage}[2]{{%
\expandafter\ifx\csname l@#1\endcsname\relax
\typeout{** WARNING: IEEEtran.bst: No hyphenation pattern has been}%
\typeout{** loaded for the language `#1'. Using the pattern for}%
\typeout{** the default language instead.}%
\else
\language=\csname l@#1\endcsname
\fi
#2}}
\providecommand{\BIBdecl}{\relax}
\BIBdecl

\bibitem{Rappaport:2013jk}
T.~S. Rappaport \emph{et~al.}, ``{Millimeter wave mobile communications for 5G
  cellular: It will work!}'' \emph{IEEE Access}, vol.~1, pp. 335--349, 2013.

\bibitem{Rangan:2014ia}
S.~Rangan, T.~S. Rappaport, and E.~Erkip, ``{Millimeter-wave cellular wireless
  networks: Potentials and challenges},'' in \emph{Proceedings of the IEEE},
  2014.

\bibitem{Singh:2015eh}
S.~Singh, M.~N. Kulkarni, A.~Ghosh, and J.~G. Andrews, ``Tractable model for
  rate in self-backhauled millimeter wave cellular networks,'' \emph{IEEE
  Journal on Selected Areas in Communications}, vol.~33, no.~10, pp.
  2196--2211, 2015.

\bibitem{ghosh2014millimeter}
A.~Ghosh \emph{et~al.}, ``Millimeter-wave enhanced local area systems: A
  high-data-rate approach for future wireless networks,'' \emph{IEEE Journal on
  Selected Areas in Communications}, vol.~32, no.~6, pp. 1152--1163, 2014.

\bibitem{Niu15}
Y.~Niu \emph{et~al.}, ``Exploiting device-to-device communications in joint
  scheduling of access and backhaul for mmwave small cells,'' \emph{IEEE
  Journal on Selected Areas in Communications}, vol.~33, no.~10, pp.
  2052--2069, 2015.

\bibitem{Zhu16}
Y.~Zhu \emph{et~al.}, ``{QoS-aware scheduling for small cell millimeter wave
  mesh backhaul},'' in \emph{2016 IEEE ICC}, 2016, pp. 1--6.

\bibitem{Feng16}
W.~Feng \emph{et~al.}, ``Millimetre-wave backhaul for 5g networks: Challenges
  and solutions,'' \emph{Sensors}, vol.~16, no.~6, p. 892, 2016.

\bibitem{Li17}
Y.~Li \emph{et~al.}, ``A joint scheduling and resource allocation scheme for
  millimeter wave heterogeneous networks,'' in \emph{WCNC}, 2017, pp. 1--6.

\bibitem{Hajek88}
B.~Hajek and G.~Sasaki, ``Link scheduling in polynomial time,'' \emph{IEEE
  Transactions on Information Theory}, vol.~34, no.~5, pp. 910--917, 1988.

\bibitem{Huang13}
P.~K. Huang, X.~Lin, and C.~C. Wang, ``A low-complexity congestion control and
  scheduling algorithm for multihop wireless networks with order-optimal
  per-flow delay,'' \emph{IEEE/ACM Transactions on Networking}, vol.~21, no.~2,
  pp. 495--508, 2013.

\bibitem{Angelakis14}
V.~Angelakis \emph{et~al.}, ``Minimum-time link scheduling for emptying
  wireless systems: Solution characterization and algorithmic framework,''
  \emph{IEEE Transactions on Information Theory}, vol.~60, no.~2, pp.
  1083--1100, 2014.

\bibitem{HariharanS12}
S.~Hariharan and N.~B. Shroff, ``On sample-path optimal dynamic scheduling for
  sum-queue minimization in trees under the k-hop interference model,'' in
  \emph{IEEE INFOCOM}, 2012, pp. 999--1007.

\bibitem{Ji16}
B.~Ji, G.~R. Gupta, and Y.~Sang, ``Node-based service-balanced scheduling for
  provably guaranteed throughput and evacuation time performance,'' in
  \emph{IEEE INFOCOM}, 2016.

\bibitem{Edmonds65b}
J.~Edmonds, ``Maximum matching and a polyhedron with $0,1$ vertices,'' \emph{J.
  of Res. the Nat. Bureau of Standards}, vol. 69~B, pp. 125--130, 1965.

\bibitem{Tang06}
J.~Tang, G.~Xue, and W.~Zhang, ``Maximum throughput and fair bandwidth
  allocation in multi-channel wireless mesh networks,'' in \emph{IEEE INFOCOM},
  2006, pp. 1--10.

\bibitem{Tassiulas02}
L.~Tassiulas and S.~Sarkar, ``Maxmin fair scheduling in wireless networks,'' in
  \emph{IEEE INFOCOM}, vol.~2, 2002, pp. 763--772.

\bibitem{Khachiyan80}
L.~Khachiyan, ``Polynomial algorithms in linear programming,'' \emph{USSR
  Computational Mathematics and Mathematical Physics}, vol.~20, no.~1, pp. 53
  -- 72, 1980.

\bibitem{Nemhauser91}
G.~L. Nemhauser and S.~Park, ``A polyhedral approach to edge coloring,''
  \emph{Oper. Res. Lett.}, vol.~10, no.~6, pp. 315--322, 1991.

\bibitem{Dantzig55}
G.~B. Dantzig \emph{et~al.}, ``The generalized simplex method for minimizing a
  linear form under linear inequality restraints,'' \emph{Pacific Journal of
  Mathematics}, vol.~5, no.~2, pp. 183--195, 1955.

\bibitem{Nakano95}
S.-i. Nakano, X.~Zhou, and T.~Nishizeki, \emph{Edge-coloring algorithms}.\hskip
  1em plus 0.5em minus 0.4em\relax Berlin, Heidelberg: Springer Berlin
  Heidelberg, 1995, pp. 172--183.

\bibitem{Holyer81}
I.~Holyer, ``The {NP}-completeness of edge-coloring,'' \emph{SIAM Journal on
  Computing}, vol.~10, no.~4, pp. 718--720, 1981.

\bibitem{Karloff87}
H.~J. Karloff and D.~B. Shmoys, ``Efficient parallel algorithms for edge
  coloring problems,'' \emph{J. Algorithms}, vol.~8, no.~1, pp. 39--52, 1987.

\bibitem{Cole01}
R.~Cole, K.~Ost, and S.~Schirra, ``{Edge-Coloring Bipartite Multigraphs in $O(E
  \log D)$ Time},'' \emph{Combinatorica}, vol.~21, no.~1, pp. 5--12, 2001.

\bibitem{Akdeniz14}
M.~R. Akdeniz \emph{et~al.}, ``Millimeter wave channel modeling and cellular
  capacity evaluation,'' \emph{IEEE Journal on Selected Areas in
  Communications}, vol.~32, no.~6, pp. 1164--1179, 2014.

\bibitem{Kolmogorov09}
V.~Kolmogorov, ``Blossom {V}: a new implementation of a minimum cost perfect
  matching algorithm,'' \emph{Mathematical Programming Computation}, vol.~1,
  no.~1, pp. 43--67, 2009.

\bibitem{Groetschel81}
M.~Gr{\"o}tschel, L.~Lov{\'a}sz, and A.~Schrijver, ``The ellipsoid method and
  its consequences in combinatorial optimization,'' \emph{Combinatorica},
  vol.~1, no.~2, pp. 169--197, 1981.

\end{thebibliography}

%\normalsize
%\input{appendix}

%}
\end{document}